%% file: main.tex
\documentclass[table]{lmcs} 
\pdfoutput=1

\usepackage{lastpage}
\lmcsdoi{18}{1}{3}
\lmcsheading{}{\pageref{LastPage}}{}{}%
{Dec.~21,~2020}{Feb.~15,~2023}{}


\input{content/preamble}

\begin{document}

\title[Good-for-games $\omega$-Pushdown Automata]{Good-for-games \boldmath$\omega$-Pushdown Automata}

\titlecomment{A preliminary version of this manuscript was presented at LICS 2020~\cite{LICS2020}. This version extends it with an analysis of resource-bounded resolvers, of the parity index hierarchy as well as of the synthesis problem.}

\author[K.~Lehtinen]{Karoliina Lehtinen}[a]	
\address{CNRS, Aix-Marseille Université  and Université de Toulon, LIS, Marseille, France}	
\email{karoliina.lehtinen@lis-lab.fr}  

\author[M.~Zimmermann]{Martin Zimmermann}[b]	
\address{University of Liverpool, Liverpool, United Kingdom \newline
  \textit{current affiliation:} Aalborg University, Aalborg, Denmark}	
\email{Martin.Zimmermann@liverpool.ac.uk}  

\begin{abstract}
  \noindent \input{content/abstract}
\end{abstract}

\maketitle

\section{Introduction}
\label{sec_intro}
\input{content/intro.tex}

\section{Preliminaries}
\label{sec_definitions}
\input{content/definitions.tex}

\section{Good-for-games Pushdown Automata}
\label{sec_gfg}
\input{content/gfg.tex}
\section{Good-for-games Pushdown Automata are Indeed Good for Games}
\label{sec_gfgaregfg}
\input{content/games.tex}

\section{Closure Properties}
\label{section_closure}
\input{content/closure.tex}

\section{Deciding Good-for-gameness}
\label{section_decisionproblems}
\input{content/decisions.tex}
\section{Resource-bounded Resolvers}
\label{section_resolver}
\input{content/resolver}

\section{The Parity Index Hierarchy}
\label{section_parityindex}
\input{content/acceptance.tex}

\section{Comparison to Visibly Pushdown Languages}
\label{section_visibly}
\input{content/visibly.tex}

\section{Conclusion}
\label{section_conc}
\input{content/conclusion}

\bibliographystyle{alphaurl}
\bibliography{biblio}

\appendix
\section{Another Contextfree Language that is not Good-for-games}
\label{section_appendix}
\input{content/appendix}

\section{Synthesis From $\omega$-GFG-CFL Specifications}
\label{section_appendixgames}
\input{content/gamesappendix}
\end{document}

%% file: content/preamble.tex
\usepackage{hyperref}

\usepackage[utf8]{inputenc}

\usepackage{amsmath}
\usepackage{amssymb}

\usepackage{lettrine}

\usepackage{xspace}
\usepackage{turnstile}

\usepackage{booktabs}
\usepackage{multirow}
\usepackage{pifont}

\usepackage{tikz}
\usetikzlibrary{decorations.pathreplacing}
\usetikzlibrary{automata}
\tikzset{>=stealth, shorten >=1pt}
\tikzset{every edge/.style = {thick, ->, draw}}
\tikzset{every loop/.style = {thick, ->, draw}}
\usetikzlibrary{positioning}
\usetikzlibrary{calc}

\newcommand{\set}[1]{\{#1\}}
\newcommand{\nats}{\mathbb{N}}
\newcommand{\ints}{\mathbb{Z}}
\renewcommand{\epsilon}{\varepsilon}
\newcommand{\size}[1]{|#1|}
\newcommand{\proj}{\pi}
\newcommand{\rev}[1]{#1^R}

\newcommand{\autud}{\mathcal{P}_{\!d}}
\newcommand{\autm}{\mathcal{M}}
\newcommand{\aut}{\mathcal{P}}
\newcommand{\daut}{\mathcal{D}}
\newcommand{\faut}{\mathcal{F}}
\newcommand{\auta}{\mathcal{A}}

\newcommand{\raut}{\mathcal{R}}

\newcommand{\autt}{\mathcal{T}}
\newcommand{\mymark}[1]{\overline{\vphantom{\Delta}#1}}

\newcommand{\initmark}{I}
\newcommand{\col}{\Omega}
\newcommand{\Gammabot}{\Gamma_{\!\bot}}

\newcommand{\sh}{\mathrm{sh}}

\newcommand{\trans}[1]{\xrightarrow{#1}}

\newcommand{\hstrat}{r}

\newcommand{\Sigmahash}{\Sigma_\#}

\newcommand{\dummy}{\text{\textvisiblespace}}

\newcommand{\gsgame}{\mathcal{G}}
\newcommand{\SigmaI}{\Sigma_1}
\newcommand{\SigmaO}{\Sigma_2}

\newcommand{\cfl}{$\omega$-CFL\xspace}
\newcommand{\dcfl}{$\omega$-DCFL\xspace}
\newcommand{\gfgcfl}{$\omega$-GFG-CFL\xspace}
\newcommand{\vpl}{$\omega$-VPL\xspace}

\newcommand{\ops}{I}
\newcommand{\zero}{\texttt{0}}
\newcommand{\plus}{\texttt{+}}
\newcommand{\minus}{\texttt{-}}
\newcommand{\el}{EL}

\newcommand{\seplang}{L_{\text{ss}}}

\newcommand{\Sigmatilde}{\widetilde{\Sigma}}
\newcommand{\Sigmacall}{\Sigma_c}
\newcommand{\Sigmaskip}{\Sigma_s}
\newcommand{\Sigmareturn}{\Sigma_r}

\newcommand{\repbddlang}{L_{\text{repbdd}}}
\newcommand{\val}{\text{val}}
\newcommand{\ext}{\text{ext}}

\newcommand{\suffix}[1]{\mathrm{sff}(#1)}
\newcommand{\prefix}[1]{\mathrm{prf}(#1)}

\newcommand{\exptime}{\textsc{ExpTime}}
\newcommand{\ptime}{\textsc{PTime}}
\newcommand{\twoexptime}{\textsc{2ExpTime}}

\newcommand{\yes}{\ding{51}}%
\newcommand{\no}{\ding{55}}%

\newcommand{\undec}{undecidable}

%% file: content/abstract.tex
We introduce good-for-games $\omega$-pushdown automata ($\omega$-GFG-PDA).
These are automata whose nondeterminism can be resolved based on the input processed so far.
Good-for-gameness enables automata to be composed with games, trees, and other automata, applications which otherwise require deterministic automata.

Our main results are that $\omega$-GFG-PDA are more expressive than deterministic $\omega$-pushdown automata and that solving infinite games with winning conditions specified by $\omega$-GFG-PDA is EXPTIME-complete. Thus, we have identified a new class of $\omega$-contextfree winning conditions for which solving games is decidable. It follows that the universality problem for $\omega$-GFG-PDA is in EXPTIME as well.

Moreover, we study closure properties of the class of languages recognized by $\omega$-GFG-PDA and decidability of good-for-gameness of $\omega$-pushdown automata and languages.
Finally, we compare $\omega$-GFG-PDA to $\omega$-visibly PDA, study the resources necessary to resolve the nondeterminism in $\omega$-GFG-PDA, and prove that the  parity index hierarchy for $\omega$-GFG-PDA is infinite.

This is a corrected version of the paper \url{https://arxiv.org/abs/2001.04392v6} published originally on January 7, 2022.

%% file: content/intro.tex
Good-for-gameness is the new determinism, and not just for solving games. 
Good-for-games automata also lend themselves to composition with other automata and trees.
These problems have in common that they are traditionally addressed with deterministic automata, which, depending on the exact type used, may be less succinct or even less expressive than nondeterministic ones.
Good-for-games automata overcome this restriction by allowing a limited form of nondeterminism that does not interfere with composition. 

As an example, consider the setting of infinite-duration two-player zero-sum games of perfect information.
In such a game, two players interact to produce a play, an infinite word over some alphabet. A winning condition specifies a partition of the set of plays indicating the winner of each play.
Here, we are concerned with games whose winning condition is explicitly given by an automaton recognizing the winning plays for one player.\footnote{This should be contrasted with the classical setting of say parity games, where the winning condition is implicitly encoded by a colouring of the arena specifying the interaction between the players.}
This setting arises, for example, when solving the LTL synthesis problem where the winning condition is specified by an LTL formula, which can be turned into an automaton.

The usual approach to solving a game with a winning condition given by an automaton~$\auta$ is to obtain an equivalent deterministic automaton~$\daut$ and then solve an arena-based game with an implicit winning condition given by the acceptance condition of $\daut$. 
The arena simultaneously captures the interaction between the players, which results in a play, and constructs the run of $\daut$ on this play.
The resulting arena-based game has the same winner as the original game with winning condition recognized by $\auta$.
For example, if the original winning condition is $\omega$-regular, there is a deterministic parity automaton recognizing it, and the resulting game is a parity game, which can be effectively solved.

However, the correctness of this reduction crucially depends on the on-the-fly construction of the run of $\daut$ on the play. 
For nondeterministic automata, one might be tempted to let the nondeterministic choices be resolved by the player who wins if the automaton accepts. However, an accepting run cannot necessarily be constructed on-the-fly, even if one exists, as the resolution of nondeterministic choices might depend on the whole play rather than on the current finite play prefix. 
A simple example is a winning condition that allows the player who resolves the nondeterminism to win the original game while her opponent wins in the arena-based game by using her nondeterministic choices against her.
This is the case in the parity automaton presented in Figure~\ref{fig:introexa}, which accepts all words, but in which nondeterminism must decide whether a word has finitely or infinitely many occurrences of $a$.

\begin{figure}[h]
    \centering
    \begin{tikzpicture}[thick]
    \tikzset{every state/.style = {minimum size =19}}

    \node[state,fill=lightgray] (i) at (0,0) {};
    \node[state] (1) at (-2,0) {$q_1$};
    \node[state,fill=lightgray] (2) at (2,0) {$q_2$};
    \node[state,fill=black] (11) at (-4,-1) {};
    \node[state] (22) at (4,-1) {};

    \path[-stealth]
     (0,.75) edge (i)
     (i) edge node[above] {$a,b$} (1)
     (i) edge node[above] {$a,b$} (2)
     (1) edge[bend left] node[below] {$a$} (11)
     (11) edge[bend left] node[above] {$b$} (1)
     (2) edge[bend right] node[below] {$a$} (22)
     (22) edge[bend right] node[above] {$b$} (2)
     (11) edge[loop left] node[left] {$a$} ()
     (1) edge[loop above] node[above] {$b$} ()
     (22) edge[loop right] node[right] {$a$} ()
     (2) edge[loop above] node[above] {$b$} ()
     ;
    \end{tikzpicture}
    
    \caption{A nondeterministic parity automaton that recognises $\{a,b\}^\omega$. A run is accepting if it visits the black vertex only finitely often and a white vertex infinitely often.}
    \label{fig:introexa}
\end{figure}
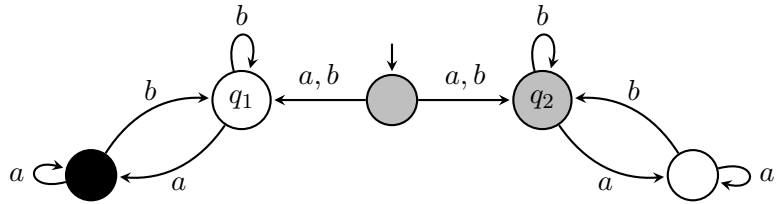

Assume that the opponent has control over the number of $a$'s that appear during a play.
This number does not affect the winner of the original game, as the automaton is universal. 
But if the initial nondeterminism is resolved by moving to $q_1$, then the opponent wins by enforcing infinitely many $a$, while if the nondeterminism is resolved by moving to $q_2$ then he wins by enforcing finitely many $a$.

Good-for-games automata (also known as history-deter\-ministic automata\footnote{The terms ``history-deterministic" and ``good-for-games" have often been used interchangeably, but recently it has come to light that history-determinism and compositionality with games do not coincide on all models~\cite{BL21}. We still prefer to use the term ``good-for-games" for consistency with previous publications, although it is open whether these two notions coincide in the setting we consider here.}~\cite{Col09}), introduced by Henzinger and Piterman~\cite{HP06}, are nondeterministic (or even alternating~\cite{Col13,BL19}) automata whose nondeterminism can be resolved based on\-ly on the input processed so far.
This property implies that the previously described procedure yields the correct winner, even if the automaton is not deterministic.

Since their introduction, Boker, Kupferman, Kuperberg and Skrzypczak have shown that $\omega$-regular good-for-games automata are also suitable for composition with trees~\cite{BKKS13}, which can be seen as the special case of one-player arena-based games, while Boker and Lehtinen have shown that good-for-games automata are suitable for automata composition in the following sense~\cite{BL19}: 
If an $\omega$-regular good-for-games automaton~$\mathcal{B}$ recognizes the set of accepting runs of an alternating $\omega$-regular automaton~$\mathcal{A}$, then the composition of $\mathcal{A}$ and $\mathcal{B}$ is an $\omega$-regular automaton that recognizes the same language as $\mathcal{A}$, but with the acceptance condition of $\mathcal{B}$. In other words, good-for-games automata, like deterministic automata, can be used both to simplify the winning conditions of games and the acceptance conditions of automata.
Kuperberg and Skrzypczak showed that $\omega$-regular good-for-games co-Büchi automata can be exponentially more succinct than deterministic ones while good-for-games Büchi automata are at most quadratically more succinct~\cite{KS15}.

However, since deterministic parity automata, which are trivially good-for-games, express all $\omega$-regular languages, succinctness is the most good-for-gameness can offer in the $\omega$-reg\-u\-lar setting.
In contrast, deterministic automata models are in general less expressive, not just less succinct, than their nondeterministic counterparts. 
This is true, for example, for pushdown automata and for various types of quantitative automata. 
We argue that in such cases, it is worthwhile to investigate good-for-games automata as an alternative to deterministic ones, as they form a potentially larger class of winning conditions for which solving games is decidable.
Indeed, the study of quantitative automata lead to the independent introduction of good-for-games automata by Colcombet. 
In particular, in the setting of regular cost functions, good-for-games cost automata are as expressive as nondeterministic ones, unlike deterministic ones~\cite{Col09}.

So far the case of $\omega$-pushdown automata has not been considered.
Here, the increased expressiveness of nondeterministic automata comes at a heavy price: games with winning conditions given by nondeterministic $\omega$-pushdown automata are undecidable~\cite{DBLP:journals/tcs/Finkel01a} while those with winning conditions given by deterministic ones are decidable~\cite{DBLP:journals/iandc/Walukiewicz01}.
Hence, in this work, we introduce and study good-for-games $\omega$-pushdown automata ($\omega$-GFG-PDA) to push the frontier of decidability for games with $\omega$-contextfree winning conditions.

\subsection*{Our contributions}
Our first results concern expressiveness:
In Section~\ref{sec_gfg}, we prove that $\omega$-GFG-PDA are strictly more expressive than deterministic $\omega$-pushdown automata ($\omega$-DPDA), but not as expressive as nondeterministic $\omega$-pushdown automata ($\omega$-PDA).
So, they do form a new class of $\omega$-context\-free languages. This is in contrast to the $\omega$-regular setting where, for each of the usual acceptance conditions, deterministic and GFG automata recognize exactly the same languages.
Also, in Section~\ref{section_parityindex}, we prove that the parity index hierarchy is infinite for $\omega$-GFG-PDA.

Second, in Section~\ref{sec_gfgaregfg}, we show that $\omega$-GFG-PDA live up to their name: 
Determining the winner of a game with a winning condition specified by an $\omega$-GFG-PDA is $\exptime$-complete, as for the special case of games with winning conditions specified by $\omega$-DPDA~\cite{DBLP:journals/iandc/Walukiewicz01}.
Furthermore, winning strategies for the player aiming to satisfy the specification are implementable by pushdown transducer, which can be computed in exponential time.
The player violating the specification does not necessarily have a winning strategy that is implementable by a pushdown transducer.
The decidability result has to be contrasted with the undecidability of games with a winning condition specified by an $\omega$-PDA~\cite{DBLP:journals/tcs/Finkel01a}. As a corollary, the universality of $\omega$-GFG-PDA is also in $\exptime$, while it is undecidable for $\omega$-PDA. 
Table~\ref{table:results:dec} sums up our results on decidability.

\begin{table*}
\centering
\renewcommand{\arraystretch}{1.3}
\renewcommand{\tabcolsep}{4pt}

\begin{tabular}{lccccc}\toprule

& Emptiness & Universality & Solving Games& \\
 \midrule
$\omega$-DPDA & \ptime & \ptime & \exptime &\cite{DBLP:journals/jcss/CohenG78,DBLP:journals/iandc/Walukiewicz01}\\
\rowcolor{lightgray}$\omega$-GFG-PDA & \ptime & \exptime$^\ast$ & \exptime & here\\
$\omega$-VPA & \ptime & \exptime & \twoexptime &\cite{AlurM04,DBLP:conf/fsttcs/LodingMS04} \\
$\omega$-PDA & \ptime\ & \undec & \undec &\cite{DBLP:journals/jcss/CohenG77,DBLP:journals/jcss/CohenG77a,DBLP:journals/tcs/Finkel01a}\\
\bottomrule
\end{tabular}
\caption{Summary of our results on decision problems (in light gray) and comparison to other classes of contextfree languages. All problems are complete for the respective complexity class unless marked with an asterisks.}
\label{table:results:dec}
\end{table*}

Third, in Section~\ref{section_closure}, we study the closure properties of $\omega$-GFG-PDA, which are almost nonexistent, and, in Section~\ref{section_decisionproblems}, prove that both the problems of deciding whether a given $\omega$-PDA is good-for games, and of deciding whether a given $\omega$-PDA is lan\-guage equivalent to an $\omega$-GFG-PDA are undecidable.
Table~\ref{table:results:closure} sums up our results on closure properties.

\begin{table*}
\centering
\renewcommand{\arraystretch}{1.3}
\renewcommand{\tabcolsep}{4pt}

\begin{tabular}{lccccc}\toprule

 & Intersection & Union & Complement & Set difference  & Homomorphism   \\
 \midrule
$\omega$-DPDA & \no & \no & \yes & \no & \no \\
\rowcolor{lightgray}$\omega$-GFG-PDA & \no & \no & \no & \no & \no\\
$\omega$-VPA & \yes & \yes & \yes & \yes & \no \\
$\omega$-PDA & \no & \yes & \no & \no & \yes \\
\bottomrule
\end{tabular}
\caption{Summary of our results on closure properties (in light gray) and comparison to other classes of contextfree languages.}
\label{table:results:closure}
\end{table*}

Fourth, in Section~\ref{section_resolver}, we study the resources necessary to resolve nondeterminism in $\omega$-GFG-PDA.
In general, pushdown transducers are not sufficient while $\omega$-GFG-PDA with finite-state machines resolving the nondeterminism are always determinizable.

Fifth, in Section~\ref{section_visibly}, we compare $\omega$-GFG-PDA with visibly pushdown automata~($\omega$-VPA)~\cite{AlurM04}, a class of $\omega$-PDA with robust closure properties and for which solving games is also decidable~\cite{DBLP:conf/fsttcs/LodingMS04}. We show that the classes of languages recognized by $\omega$-GFG-PDA and $\omega$-VPA are incomparable with respect to inclusion.
See Figure~\ref{figure:classes} for an overview of our results on the relations between these classes.

\begin{figure}[b]
    \centering
\begin{tikzpicture}[scale=.95,ultra thick]

\draw[rounded corners,fill=gray!5] (-1,-2) rectangle (8.5,2.25);
 \node[anchor=west] at (-.9,-1.625) {$\omega$-PDA};

\draw[rounded corners,fill=gray!20] (-.5,-1.25) rectangle (6.5,1.25);
\node[anchor=west] at (-.4,-.875) {$\omega$-GFG-PDA};

\draw[rounded corners,fill=gray!35] (0,-.5) rectangle (6,.5);
\node[anchor=west] at (0.1,0) {$\omega$-DPDA};

\draw[rounded corners,fill=white, fill opacity=.8] (5,-1.625) rectangle (8.25,2);
 \node[anchor=east] at (8.15,-1.25) {$\omega$-VPA};

 \node[draw,fill=black,circle,inner sep = .05cm,label=right:Theorem~\ref{thm_gfgvscfl}] at (-.5,1.625) {};

\node[draw,fill=black,circle,inner sep = .05cm,label=right:Theorems~\ref{thm_dcflvsgfgcfl} and \ref{thm_vpl}] at (0,.875) {};

\node[draw,fill=black,circle,inner sep = .05cm,label=right:Theorem~\ref{thm_vpl}] at (5.5,1.625) {};

\end{tikzpicture}
\caption{The classes of $\omega$-languages recognized by the automata considered in this work. Languages separating these classes are shown as black dots.}
\label{figure:classes}
\end{figure}
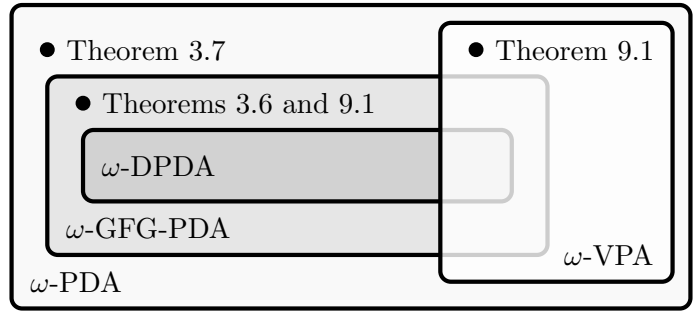

%% file: content/definitions.tex
An alphabet~$\Sigma$ is a finite nonempty set of letters. The set of finite words over $\Sigma$ is denoted by $\Sigma^*$, the set of nonempty finite words over $\Sigma$ by $\Sigma^+$, and the set of infinite words over $\Sigma$ by $\Sigma^\omega$. 
The  empty word is denoted by $\epsilon$, the length of a finite word~$v$ is denoted by $\size{v}$, and the $n$-th letter of a finite or infinite word is denoted by $w(n)$ (starting with $n = 0$).
An $\omega$-language over $\Sigma$ is a subset of $\Sigma^\omega$.

For alphabets~$\Sigma_1,\Sigma_2$, we extend functions~$f \colon \Sigma_1 \rightarrow \Sigma_2^*$ homomorphically to finite and infinite words over $\Sigma_1$ via
\[f(w) = f(w(0)) f(w(1)) f(w(2)) \cdots .\]
For example, if $\Sigma_1 = \Sigma \times \Sigma'$, then $\proj_i(a_1, a_2) \mapsto a_i$ for $(a_1,a_2)\in \Sigma_1$ denotes the  projection to the $i$-th component ($i \in \set{1,2}$).

An $\omega$-pushdown automaton ($\omega$-PDA for short)~$\aut = (Q, \Sigma, \Gamma, q_\initmark, \Delta, \col)$
consists of a finite set~$Q$ of states with the initial state~$q_\initmark \in Q$, an input alphabet~$\Sigma$, a stack alphabet~$\Gamma$, a transition relation~$\Delta$ to be specified, and a coloring~$\col \colon \Delta \rightarrow \nats$.
For notational convenience, we define $\Sigma_\epsilon = \Sigma \cup \set{\epsilon}$ and $\Gammabot = \Gamma \cup \set{\bot}$, where $\bot \notin \Gamma$ is a designated stack bottom symbol.
Then, the transition relation~$\Delta$ is a subset of  $Q \times \Gammabot \times \Sigma_\epsilon \times Q \times \Gammabot^{\le 2}$ that we require to neither write nor delete the stack bottom symbol from the stack:
If
$(q, \bot, a, q', \gamma) \in \Delta$, then $\gamma \in \bot \cdot (\Gamma \cup \set{\epsilon}) $, and if $(q, X, a, q', \gamma) \in \Delta$ for $X \in \Gamma$, then $\gamma \in \Gamma^{\le 2}$. 
Given a transition~$\tau = (q,X,a,q',\gamma)$ let $\ell(\tau) = a \in \Sigma_\epsilon$. 
We say that $\tau$ is an $\ell(\tau)$-transition and that $\tau$ is a $\Sigma$-transition, if $\ell(\tau) \in \Sigma$.
For a finite or infinite sequence~$\rho$ over $\Delta$, $\ell(\rho)$ is defined by applying $\ell$ homomorphically to every transition.

A stack content is a finite word in $\bot \Gamma^*$ (i.e., the top of the stack is at the end) and a configuration~$c = (q, \gamma)$ of $\aut$ consists of a state~$q \in Q$ and a stack content~$\gamma$.  
The stack height of $c$ is $\sh(c) = \size{\gamma}-1$.
The initial configuration is $(q_\initmark, \bot)$.

A transition~$\tau = (q, X, a, q', \gamma' ) \in \Delta $ is enabled in a configuration~$c$ if $c = (q, \gamma X)$ for some $\gamma \in \Gammabot^*$.
In this case, we write~$(q, \gamma X) \trans{\tau} (q', \gamma\gamma')$.
A run of $\aut$ is an infinite sequence~$\rho = c_0 \tau_0 c_1 \tau_1 c_2 \tau_2 \cdots$ of configurations and transitions with $c_0$ being the initial configuration and $c_n \trans{\tau_n}c_{n+1}$ for every $n$.
Finite run prefixes are defined analogously and are required to end in a configuration.
The infinite run~$\rho$ is a run of $\aut$ on $w \in \Sigma^\omega$, if $w = \ell(\tau_0 \tau_1 \tau_2 \cdots) $ (this implies that $\rho$ contains infinitely many $\Sigma$-transitions). 
We say that $\rho$ is accepting if $\limsup_{n \rightarrow \infty} \col(\tau_n)$ is even, i.e., if the maximal color labeling infinitely many transitions is even.
The language~$L(\aut)$ recognized by $\aut$ contains all $w \in \Sigma^\omega$ such that $\aut$ has an accepting run on $w$. 

\begin{rem}
Let $ c_0 \tau_0 c_1 \tau_1 c_2 \tau_2 \cdots$ be a run of $\aut$.
Then, the sequence~$c_0c_1 c_2\cdots$ of configurations is uniquely determined by the sequence~$\tau_0\tau_1 \tau_2\cdots$ of transitions.
Hence, whenever convenient, we treat a sequence of transitions as a run if it indeed induces one (not every such sequence does induce a run, e.g., if a transition~$\tau_n$ is not enabled in $c_n$).
\end{rem}

We say that an $\omega$-PDA~$\aut$ is deterministic if 
\begin{itemize}
    \item for every $q \in Q$, every $X \in \Gammabot$, and every $a \in \Sigma_\epsilon$, there is at most one transition of the form~$(q, X, a, q', \gamma) \in \Delta$ for some $q'$ and some $\gamma$, and
    
    \item for every $q \in Q$ and every $X \in \Gammabot$, if there is a transition~$(q, X, \epsilon, q_1, \gamma_1) \in \Delta$ for some $q_1$ and some $\gamma_1$, then there is no $a \in \Sigma$ such that there is a transition~$(q, X, a, q_2, \gamma_2) \in \Delta$ for some $q_2$ and some $\gamma_2$.
    
\end{itemize}
As expected, a deterministic $\omega$-pushdown automaton ($\omega$-DPDA) has at most one run on every $\omega$-word. 

The class of $\omega$-languages recognized by $\omega$-PDA is denoted by \cfl and the class of $\omega$-languages recognized by $\omega$-DPDA by \dcfl.
Cohen and Gold showed that \dcfl is a strict subset of \cfl~\cite[Theorem~3.2]{DBLP:journals/jcss/CohenG78}.%
\footnote{Formally, Cohen and Gold considered automata with state-based Muller acceptance while we consider, for technical convenience, automata with transition-based parity acceptance. However, using \emph{latest appearance records}~(see, e.g.,~\cite{GTW02}) shows that both definitions are equivalent.}

\begin{exa}
\label{example:pda}
The $\omega$-PDA~$\aut$ depicted in Figure~\ref{fig:pdaexample} recognizes the $\omega$-language
\[
\set{a c^nd^n \#^\omega \mid n \ge 1} \cup \set{b c^nd^{2n} \#^\omega \mid n \ge 1}.
\]
Note that while $\aut$ is nondeterministic, $L(\aut)$ is in \dcfl.
\begin{figure}
    \centering

\begin{tikzpicture}[thick]
\def\y{.9}
\def\x{2}
\tikzset{every state/.style = {minimum size =22}}
\node[state,fill=lightgray] (i) at (0*\x,0) {$q_0$};
\node[state,fill=lightgray] (u) at (2*\x, 0) {$q_1$};
\node[state,fill=lightgray] (ad) at (4*\x,\y) {$q_2$};
\node[state,fill=lightgray] (bd1) at (4*\x,-\y) {$q_3$};
\node[state,fill=lightgray] (bd2) at (6*\x,-\y) {$q_5$};
\node[state] (acc) at (6*\x,\y) {$q_4$};

\path[-stealth]
(0,-.75) edge (i)
(i) edge node[above] {$a,\bot\mid \bot A$} (u)
(i) edge node[below] {$b,\bot\mid \bot B$} (u)
(u) edge[loop right] node[right] {$c, X\mid XN $} ()
(u) edge[bend left] node[above,yshift=.1cm] {$d, N \mid \epsilon$} (ad)
(u) edge[bend right] node[below,yshift=-.1cm] {$d, N \mid N$} (bd1)
(ad) edge[loop above] node[above] {$d, N \mid \epsilon $} ()
(bd1) edge[bend left=15] node[above] {$d, N \mid \epsilon$} (bd2)
(bd2) edge[bend left=15] node[below] {$d, N \mid N$} (bd1)
(bd2) edge node[right] {$\#, B \mid \epsilon$} (acc)
(ad) edge node[above] {$\#, A \mid \epsilon$} (acc)
(acc) edge[loop above] node[above] {$\#, \bot \mid \bot$} ()
;
\end{tikzpicture}

    \caption{The $\omega$-PDA $\aut$ from Example~\ref{example:pda}. The self-loop at the white state~$q_4$ has color~$2$ while all other transitions have color~$1$, and $X$ represents an arbitrary stack symbol from the stack alphabet~$\set{A,B,N}$ (not including $\bot$).}
    \label{fig:pdaexample}
\end{figure}
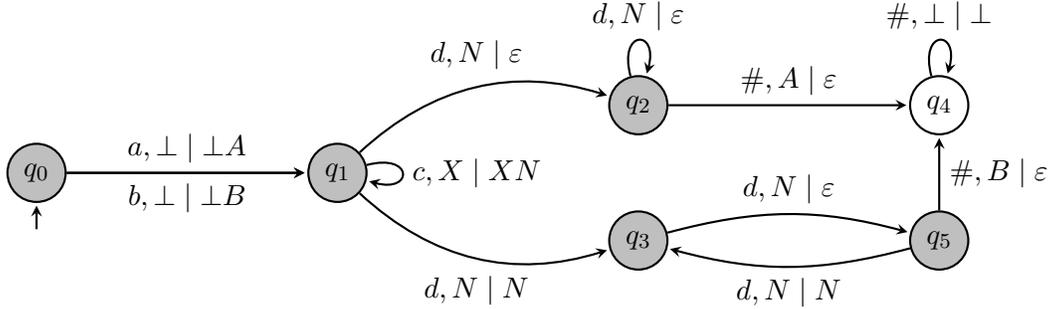

\end{exa}

%% file: content/gfg.tex
Here, we introduce good-for-games $\omega$-push\-down automata ($\omega$-GFG-PDA for short), nondeterministic $\omega$-push\-down automata whose nondeterminism can be resolved based on the run prefix constructed thus far and on the next input letter to be processed, but independently of the continuation of the input beyond the next letter.

As an example, consider the $\omega$-PDA $\aut$ from Example~\ref{example:pda}. 
It is nondeterministic, but the nondeterminism in a configuration of the form~$(q_1,\gamma N)$ can be resolved based on whether the first transition of the run processed an $a$ or a $b$:
in the former case, take the transition to state~$q_2$, in the latter case the transition to state~$q_3$. 
Afterwards, there are no nondeterministic choices to make and the resulting run is accepting whenever the input is in the language. This automaton is therefore good-for-games. The language of this automaton is clearly \dcfl; an automaton separating \dcfl and the class of languages recognised by $\omega$-GFG-PDA is one of the main results of this section, presented below.

Formally, we say that an $\omega$-PDA~$\aut = (Q, \Sigma, \Gamma, q_\initmark, \Delta, \col)$ is good-for-games, if there is a (nondeterminism) resolver for $\aut$, a function~$\hstrat \colon \Delta^* \times \Sigma \rightarrow \Delta$ such that for every $w \in L(\aut)$ the sequence~$\tau_0 \tau_1 \tau_2 \cdots \in \Delta^\omega$  defined by
\[
\tau_n = \hstrat(\tau_0 \cdots \tau_{n-1}, w(\size{ \ell(\tau_0 \cdots \tau_{n-1})}))
\]
induces an accepting run of $\aut$ on $w$.
Note that the prefix processed so far can be recovered from $\hstrat$'s input, i.e., it is $\ell(\tau_0 \cdots \tau_{n-1})$. However, the converse is not true due to the existence of $\epsilon$-transitions.
This is the reason the run prefix and not the input prefix is the input for the resolver.

\begin{rem}
\label{remark_runs}
Let $\aut$ be an $\omega$-GFG-PDA over $\Sigma$ such that every finite word~$v \in \Sigma^*$ is a prefix of some word in $L(\aut)$, and let $\hstrat$ be a resolver for $\aut$.

Then, $\hstrat$ induces a run on every input~$w \in \Sigma^\omega$, even if $w$ is not accepted by $\aut$.
Indeed, if there is no run on $w$, because either there is no enabled transition that processes the next input letter or because it ends in an infinite tail of $\epsilon$-transitions, then some word in $L(\aut)$ has no accepting run induced by $\hstrat$, which is a contradiction.
\end{rem}

Each deterministic automaton is trivially good-for-games.
We denote the class of $\omega$-languages recognized by $\omega$-GFG-PDA by \gfgcfl.
By definition, the three classes of lan\-guages we consider form a hierarchy.

\begin{prop}
\label{prop_inclusions}
\dcfl $\subseteq $ \gfgcfl $\subseteq$ \cfl.
\end{prop}

We show that both inclusions are strict. 
In particular, $\omega$-GFG-PDA are more expressive than $\omega$-DPDA.

Let $\ops = \set{\zero,\plus, \minus}$ and define the energy level~$\el(v) \in \ints$ of finite words~$v$ over $ \ops$ inductively as $\el(\epsilon) = 0$, and $\el(v\zero) = \el(v)$, $\el(v\plus) = \el(v) + 1$, as well as $\el(v\minus) = \el(v) - 1$.
We say that a word~$w \in \ops^\omega$ is \emph{safe} if $\el(w(0) \cdots w(n)) \ge 0$ for every $n \ge 0$.

\begin{rem}
\label{remark_safe}
Let $w \in \ops^\omega$.
\begin{enumerate}
    
    \item\label{remark_safe_lb} $w$ has a safe suffix if and only if there is an $s \in \nats$ such that $\el(w(0) \cdots w(n)) \ge -s$ for all $n$.

    \item\label{remark_safe_suffix} If $w$ is safe then there is an $n > 0$ such that $w(n) w(n+1) w(n+2) \cdots $ is safe as well. 
    
\end{enumerate}
\end{rem}

We fix $\Sigma = \ops \times \ops$ and define $\seplang$ to be the language containing all $w \in \Sigma^\omega$ such that $\proj_i(w)$ has a safe suffix, for some $i\in \set{1,2}$. 

\begin{lem}
\label{lemma_Lisgfg}
$\seplang  \in$ \gfgcfl.
\end{lem}

\begin{proof}
Consider the $\omega$-PDA~$\aut$ in Figure~\ref{fig:elaut} with two states~$1$ and $2$ signifying whether a potential safe suffix is being tracked in the first or second component of the input, and a single stack symbol~$N$ used to track the energy level of such a suffix. 
The initial state is arbitrary; we fix it to be $1$.

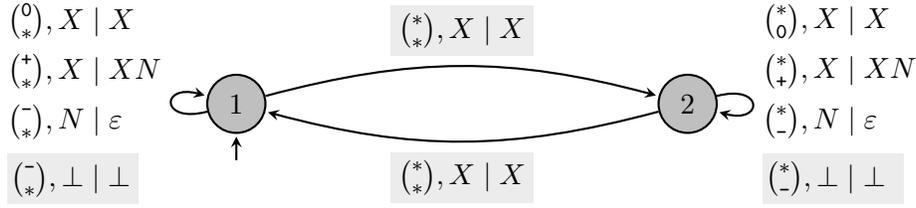
\begin{figure}
    \centering

\begin{tikzpicture}[thick]
\def\y{1.2}
\def\x{2}
\tikzset{every state/.style = {minimum size =22}}
\node[state,fill=lightgray] (l) at (0*\x,0) {$1$};
\node[state,fill=lightgray] (r) at (3*\x, 0) {$2$};

\path[-stealth]
(0,-.75) edge (l)
(l) edge[bend left=15] node[above] {\colorbox{lightgray!30}{$\binom{\ast}{\ast} ,X\mid X$}} (r)
(r) edge[bend left=15] node[below] {\colorbox{lightgray!30}{$\binom{\ast}{\ast},X\mid X$}} (l)
(l) edge[loop left] node[left,align=left] {$\binom{\zero}{\ast}, X\mid X$\\[1ex] $\binom{\plus}{\ast}, X\mid XN $\\[1ex] $\binom{\minus}{\ast}, N \mid \epsilon$\\ [1ex]
\colorbox{lightgray!30}{\!$\binom{\minus}{\ast}, \bot \mid \bot$}} ()
(r) edge[loop right] node[right,align=left] {
 $\binom{\ast}{\zero}, X\mid X$\\[1ex]
 $\binom{\ast}{\plus}, X\mid XN$ \\[1ex]
 $\binom{\ast}{\minus}, N \mid \epsilon$\\ [1ex]
\colorbox{lightgray!30}{\!$\binom{\ast}{\minus}, \bot \mid \bot$
}} ()
;
\end{tikzpicture}

    \caption{The $\omega$-PDA $\aut$ recognizing $\seplang$. Here, $\ast$ is an arbitrary letter in $\set{\zero,\plus,\minus}$ and $X$ an arbitrary stack symbol from the stack alphabet~$\set{N}$ or $\bot$. Transitions with light gray background have color $1$, all others have color~$0$.}
    \label{fig:elaut}
\end{figure}

The automaton can, at any moment, nondeterministically change its state from $1$ to $2$ and vice versa without changing the stack content (while processing an input letter that is just ignored). 
When not changing its state, say while staying in state~$i$, $\aut$ deterministically processes the next input letter~$\binom{a_1}{a_2}$. 
If $a_i = \zero$ then the stack is left unchanged and if $a_i = \plus$ then an $N$ is pushed onto the stack. 
If $a_i = \minus$ and the stack is nonempty, then the topmost $N$ is popped from the stack. 
The stack is left unchanged if $a_i = \minus$ and the stack is empty.
Note that the only nondeterministic choice is to change the state, i.e., if the state is not changed then the automaton has exactly one transition that can process the next input letter.
Furthermore, $\aut$ has no $\epsilon$-transitions, i.e., each transition processes an input letter.

A transition has color $1$ if it either changes the state, or  if it processes  a $\minus$ in the $i$-th component while in state~$i$ with an empty stack.
All other transitions have color~$0$. 
Hence, a run is accepting if and only if it eventually stays in some state~$i$ and then does not process a $\minus$ in component~$i$ while the stack is empty.

First, we show that $\seplang  = L(\aut)$, then that $\aut$ is good-for-games.
Let $w \in \seplang $, i.e., there is an $i \in \set{1,2}$ such that $\proj_i(w) = a(0) a(1) a(2) \cdots$ has a safe suffix, say $a(n) a(n+1) a(n+2) \cdots$. 
Due to Remark~\ref{remark_safe}(\ref{remark_safe_suffix}), we assume w.l.o.g.~$n>0$.
Consider the unique run~$\rho = c_0 \tau_0 c_1 \tau_1 c_2 \tau_2 \cdots$ of $\aut$ on $w$ that immediately switches to state~$i$ (if necessary) and otherwise always executes the unique enabled transition that processes the next input letter without changing state.
We have $\ell(\tau_0 \cdots \tau_{j-1}) = w(0) \cdots w(j-1)$ due to the absence of $\epsilon$-transitions.

An induction shows the invariant~$\sh(c_{n+j}) = \sh(c_n) + \el( a(n) \cdots a(n+j-1))$ for every $j \ge 0$.
Here we use the assumption $n>0$, which ensures that no transitions changing the state are used to process the safe suffix.
Hence, the state is equal to $i$ at all but possibly the first configuration and the invariant implies that no $\minus$ is processed in component~$i$ by a transition~$\tau_{n+j}$ while the stack is empty.
This implies that the run is accepting, i.e., $w \in L(\aut)$.

Conversely, let $w \in L(\aut)$. Then, there is an accepting run~$c_0 \tau_0 c_1 \tau_1 c_2 \tau_2 \cdots$ of $\aut$ on $w$.
Again, as $\aut$ has no $\epsilon$-tran\-si\-tions,  $\ell(\tau_0 \cdots \tau_{j-1}) = w(0) \cdots w(j-1)$.
Now, let $\tau_{n-1}$ be the last transition with color~$1$ (pick $\tau_{n-1} = \tau_0$ if there is no such transition) and let $c_{n} = (i,\bot N^s)$.
We define $a(n) a(n+1) a(n+2) \cdots $ to be the suffix of $\proj_i(w) = a(0) a(1) a(2) \cdots$ starting at position~$n$.

By the choice of $n$, after $\tau_{n-1}$ no transition changes the state (it is always equal to $i$) and there is no $j \ge 0 $ such that $a(n+j)=\minus$ and $\sh(c_{n+j})=0$.
Thus, an induction shows that $\el(a(n) \cdots a(n+j)) = \sh(c_{n+j}) - s $.
Hence, the energy level of the prefixes of $a(n) a(n+1) a(n+2) \cdots$ is bounded from below by $-s$. 
Thus, Remark~\ref{remark_safe}(\ref{remark_safe_lb}) implies that $a(n) a(n+1) a(n+2) \cdots$ has a safe suffix, i.e., $w \in \seplang $.

It remains to show that $\aut$ is good-for-games.\footnote{Although the language we use here is different, this argument is similar to the one used by Kuperberg and Skrzypczak~\cite{KS15} to show that good-for-games co-Büchi automata are exponentially more succinct than deterministic ones.}
Intuitively, we construct a resolver~$\hstrat$ that always chooses to track the input component that has the longest suffix that can still be extended to an infinite safe word (preferring the first component in case of a tie).  If one of the components has a safe suffix, eventually the resolver will choose to track it. The resulting run will be accepting.

Fix $v = \binom{a_1(0)}{a_2(0)} \cdots \binom{a_1(n)}{a_2(n)} \in \Sigma^+$ and let $S_i$ contain those $j \le n$ such that $a_i(j) \cdots a_i(n)$ is safe (which is defined as expected).
Further, let 
\[i(v) = \begin{cases}
1 &\text{if $\min S_1 \le \min S_2$,}\\
2 &\text{otherwise,}
\end{cases}
\]
where we use $\min \emptyset = n+1$. 
We define $\hstrat(\tau_0 \cdots \tau_{n-1},a)$ inductively to always ensure that its target state is equal to $ i( \ell(\tau_0 \cdots \tau_{n-1})a)$ and, if this does not require a state change, then the unique transition processing the $i( \ell(\tau_0 \cdots \tau_{n-1})a)$-th component of $a$ is returned.

Now, let $w \in \seplang $, i.e., there is an $i \in \set{1,2}$ such that $\proj_i(w)$ has a safe suffix. 
Then, $\hstrat$ produces a run that tracks the safe suffix that starts as early as possible (again favoring the first component in case of a tie).
As in the argument above one can show that this run is accepting, as it switches states only finitely often and processes only finitely many $-$ while the stack is empty: After the last state change, the current suffix in the tracked component might not be safe, but it is the suffix of a safe suffix. Thus, due to Remark~\ref{remark_safe}(\ref{remark_safe_suffix}), the current suffix has a safe suffix as well. Hence, from this point onward, no $-$ is processed while the stack is empty.
Thus, $\hstrat$ has the desired properties and $\aut$ is good-for-games.
\end{proof}

After having shown that $\seplang $ is in \gfgcfl, we show that it is not in \dcfl, thereby separating \dcfl and \gfgcfl.

\begin{lem}
\label{lemma_Lisnotdcfl}
$\seplang  \notin$ \dcfl.
\end{lem}

\begin{proof}
We assume towards a contradiction that there is an $\omega$-DPDA~$\aut = (Q, \Sigma, \Gamma, q_\initmark, \Delta, \col)$ recognizing $\seplang $.
Define $x_1 = \binom{\plus}{\zero} \binom{\plus}{\minus}$ and $x_2 = \binom{\zero}{\plus} \binom{\minus}{\plus}$, i.e., in $x_i$ the energy level in component~$i$ is increased by two while it is decreased by one in the other component.

Define\label{bla}
\[w_{\overline{ss}} =
x_1 \, 
(x_2)^3 \,
(x_1)^7 \,
(x_2)^{15} \,
(x_1)^{31} \,
(x_2)^{63} \,
\cdots.
\]
An induction shows \[\el( \proj_1(  x_1 (x_2)^3 \cdots (x_2)^{2^{2j}-1} ) ) = -j \] for every $j > 1$ and \[\el( \proj_2(  x_1 (x_2)^3 \cdots (x_1)^{2^{2j-1}-1} ) ) = -j \] for every $j > 0$. 
Hence, due to Remark~\ref{remark_safe}(\ref{remark_safe_lb}), $w_{\overline{ss}} \notin \seplang $.

Due to Remark~\ref{remark_runs}, $\aut$ has a run~$\rho = c_0 \tau_0 c_1 \tau_1 c_2 \tau_2 \cdots $ on $w_{\overline{ss}}$, which is rejecting. 
\label{stepsdef}A step of $\rho$ is a position~$n$ such that $\sh(c_n) \le \sh(c_{n+j})$ for all $j \ge 0$.
Every infinite run has infinitely many steps.
Hence, we can find two steps~$s < s'$ satisfying the following properties:

\begin{enumerate}
    
    \item There is a state~$q \in Q$ and a stack symbol~$X \in \Gammabot$ such that $c_s = (q, \gamma X)$ and $c_{s'} = (q, \gamma' X)$ for some $\gamma, \gamma'$, i.e., both configurations have the same state and topmost stack symbol.
    
    \item The maximal color labeling the sequence~$\tau_{s} \cdots \tau_{s'-1}$ of transitions leading from $c_s$ to $c_{s'}$ is odd.

    \item The sequence~$\tau_{s} \cdots \tau_{s'-1}$  processes an infix~$v$ of $w$ with $\el(\proj_i(v)) > 0$, for some $i \in \set{1,2}$. This can be achieved, as every infix of length at least $3$ of a word built by concatenating copies of the $x_i$ has a strictly positive energy level in one component (note that the infix may start or end within an $x_i$).
    
\end{enumerate}

Consider the sequence~$\tau_0 \cdots \tau_{s-1}( \tau_s \cdots \tau_{s'-1} )^\omega $ of transitions.
Due to the first property, it induces a run~$\rho'$ of $\aut$, which is rejecting due to the second property. 
Finally, due to the third property, $\rho'$ processes a word with suffix~$v^\omega$.
Such a word has a safe suffix in component~$i$, as $\el(\proj_i(v)) > 0$.
Hence, we have constructed a word in $\seplang $ such that the unique (as $\aut$ is deterministic) run of $\aut$ on $w$ is rejecting, obtaining the desired contradiction to $L(\aut) = \seplang $.
\end{proof}

Our main result of this section is now a direct consequence of the previous two lemmata:
$\omega$-GFG-PDA are more expressive than $\omega$-DPDA.

\begin{thm}
\label{thm_dcflvsgfgcfl}
\dcfl $\subsetneq$ \gfgcfl.
\end{thm}

The next obvious question is whether every (nondeterministic) $\omega$-contextfree language is good-for-games. 
Not unexpectedly, this is not the case.
The intuitive reason is that good-for-games automata allow to resolve nondeterminism based on the history of a run, but still cannot resolve nondeterminism based on the continuation of the input.
Considering a language that requires nondeterministic choices about the continuation of the input yields the desired separation.  
To this end, we adapt the classical proof that 
\[
\set{ a^nb^n \mid n\ge 1} \cup \set{ a^n b^{2n} \mid n\ge 1}\]
is not recognizable by a DPDA over finite words to our setting. 
Recall that this proof proceeds by contradiction as follows: Consider a DPDA~$\aut$ recognizing this language and inputs~$w_0 = a^nb^n$ and $w_1 = a^nb^{2n}$, both in the language. The run of $\aut$ on $w_0$ is a prefix of the run~$\rho_1$ of $\aut$ on $w_1$ and ends an accepting state. Furthermore, no other prefix of $\rho_1$ can be accepting.

So, one can modify $\aut$ so that it behaves like $\aut$ until an accepting state is reached for the first time (the run thus far has processed an input~$a^nb^n$ for some $n$ up to that point due to the prefix property above). From that point onward, the modified automaton processes $c$'s instead of $b$'s (but behaves like $\aut$ otherwise) until it visits an accepting state again, i.e., instead of processing another $n$ $b$'s it processes $n$ $c$'s. Hence, the modified automaton accepts the non-contextfree language~$\set{a^nb^nc^n \mid n\ge 1}$, which yields the desired contradiction. 

Note that the prefix property of runs underlying the argument above is also satisfied by good-for-games automata.
For technical convenience, we use the language
\[L = 
\set{ (a\#)^n (b\#)^n \#^\omega \mid n\ge 1} \cup \set{ (a\#)^n (b\#)^{2n} \#^\omega \mid n\ge 1},
\]
i.e., we add the dummy symbols~$\#$ which obfuscate the next (proper) letter to be processed from the resolver.

\begin{thm}
\label{thm_gfgvscfl}
\gfgcfl $\subsetneq$ \cfl.
\end{thm}

\begin{proof}
The language~$L$ is in \cfl, as an $\omega$-PDA can process a prefix~$(a\#)^n$ by storing $n$ in unary on the stack and then nondeterministically guess and verify whether the remaining suffix is  $(b\#)^n \#^\omega$ (by popping a symbol from the stack for every $b$) or whether it is $(b\#)^{2n} \#^\omega$ (by popping a symbol from the stack for every other $b$), similarly to the automaton from Example~\ref{example:pda}

We claim that $L$ is not in \gfgcfl.
We assume towards a contradiction that there is an $\omega$-GFG-PDA~$\aut = (Q, \Sigma, \Gamma, q_\initmark, \Delta, \col)$ with $L(\aut) = L$, say with resolver~$\hstrat$. In what follows, we will reach a contradiction by constructing an $\omega$-PDA that recognizes the language 
\[
L_{abc} = \set{ (a\#)^n (b\#)^n (c\#)^n \#^\omega \mid n \ge 1 },
\]
which is not in \cfl.
This follows from $\set{ a^n b^n c^n \mid n \ge 1 }$ not being contextfree and from the closure properties of \cfl~\cite{DBLP:journals/jcss/CohenG77a} and of the contextfree languages~\cite{DBLP:books/daglib/0016921}.

First, we note that the language
\[
C = \set{
\gamma q \in \bot\Gamma^*Q \mid \aut \text{ accepts } \#^\omega \text{ when starting in } (q,\gamma) }
\]
is regular.
This can be shown by first noting that 
\begin{align*}
C_0 = \{
\gamma X q \in \bot\Gamma^*Q \mid \aut \text{ accepts } \#^\omega \text{ when starting in }  (q,\gamma X) \text{ with a run that only}\\ \text{visits configurations} \text{ of stack height greater or equal to } \sh(q,\gamma X) \}
\end{align*}
is a finite union of languages~$\Gamma^*Xq \cap \bot\Gamma^*q $ for some $X \in \Gammabot$ and some $q \in Q$, and therefore regular.
Now, $C$ is equal to
\begin{align*}
\{ \gamma q \in \bot\Gamma^*Q \mid& \text{ there is a run infix~$\rho$ with $\ell(\rho) \in \#^*$ leading from $(q,\gamma)$  to $C_0$} \}.    
\end{align*}
An application of standard saturation techniques~\cite{Buechi1964}\footnote{Also, see the survey of Carayol and Hague~\cite{CarayolHague14} for more details.} (applied to the restriction of $\aut$ to transitions labeled by $\#$ or $\epsilon$) shows that the latter set is regular, as the target set~$C_0$ is regular. 

Using a deterministic finite automaton~$\auta = (Q', \Gammabot\cup Q, q_\initmark', \delta, F)$ recognizing $C$ we construct an $\omega$-PDA~$\aut'$ as follows (also, see Figure~\ref{fig:construction}):
We extend the stack alphabet~$\Gamma$ of $\aut$ to $\Gamma \times Q'$ and define the transition relation of $\aut'$ so that it simulates a run of $\aut$ and keeps track of the state of $\auta$ reached by processing the stack content, i.e., if $\aut'$ reaches a stack content
\[
\bot (X_1, q_1) \cdots (X_s, q_s)
\]
then we have $q_j = \delta^*(q_\initmark', \bot X_0 \cdots X_j)$ for every $0 \le j \le s$.\footnote{Note that the simplest way to implement this is to replace each transition that swaps the topmost stack symbol from $X$ to $X'$ by two transitions, the first popping $X$ from the stack, and the second pushing $X'$ onto the stack (using a fresh state reached between the new transitions). }

Additionally, the states of $\aut'$ have a Boolean flag that is changed if $\aut'$ leaves a configuration of the form~$(q,\gamma)$ with $\gamma q \in C$ for the first time. 
As the stack symbols of $\aut'$ encode the run of $\auta$ on the stack content, this can be checked easily.
From there on, $\aut'$ continues to simulate a run of $\aut$, but now every transition labeled by a $b$ in $\aut$ is labeled by a $c$.

Finally, we define the acceptance condition of $\aut'$ such that it only accepts if it switches the flag, afterwards continues the simulation of a run of $\aut$ that is accepting, and processes at least one $c$ after the switch. 
Note that this requires adding a second Boolean flag to the states to check whether a $c$ has been processed.
We claim that $\aut'$ recognizes the language $L_{abc}$.

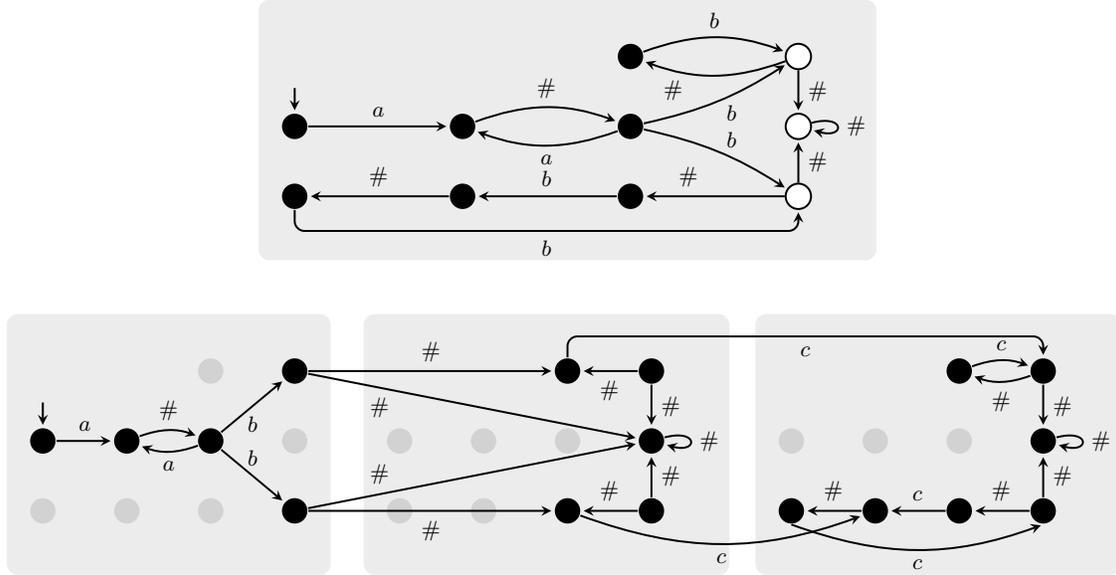
\begin{figure}
\centering
\scalebox{.93}{
    \begin{tikzpicture}[thick]

\def\x{2.4}
\def\y{1.5}
\def\d{5.6}
\tikzset{sn/.style={fill=black,circle}}
\tikzset{sng/.style={fill=lightgray!70,circle}}
\tikzset{snw/.style={fill=white,circle,draw= black}}

\coordinate (zero) at (3.6,4.5);

\draw[rounded corners,fill=lightgray!30,lightgray!30] ($ (zero) + (-.5,-\y-.4) $) rectangle ($(zero) + (1.1+ 3*\x, \y+.3)$);

\node[sn] (io) at (zero) {};
\node[sn] (1o) at ($(zero) + (1*\x,0)$) {};
\node[sn] (2o) at ($(zero) + (2*\x,0)$) {};
\node[snw] (3o) at ($(zero) + (3*\x,1)$) {};
\node[snw] (4o) at ($(zero) + (3*\x,0)$) {};
\node[snw] (5o) at ($(zero) + (3*\x,-1)$) {};
\node[sn] (6o) at ($(zero) + (2*\x,1)$) {};
\node[sn] (7o) at ($(zero) + (2*\x,-1)$) {};
\node[sn] (8o) at ($(zero) + (1*\x,-1)$) {};
\node[sn] (9o) at ($(zero) + (0*\x,-1)$) {};

\def\x{1.2}
\def\y{1.5}
\def\d{5.6}

\coordinate (zero) at (0,0);

\draw[rounded corners,fill=lightgray!30,lightgray!30] ($ (zero) + (-.5,-\y-.4) $) rectangle ($(zero) + (.5+ 3*\x, \y+.3)$);

\node[sn] (i) at (zero) {};
\node[sn] (1) at ($(zero) + (1*\x,0)$) {};
\node[sn] (2) at ($(zero) + (2*\x,0)$) {};
\node[sn] (3) at ($(zero) + (3*\x,1)$) {};
\node[sng] (4) at ($(zero) + (3*\x,0)$) {};
\node[sn] (5) at ($(zero) + (3*\x,-1)$) {};
\node[sng] (6) at ($(zero) + (2*\x,1)$) {};
\node[sng] (7) at ($(zero) + (2*\x,-1)$) {};
\node[sng] (8) at ($(zero) + (1*\x,-1)$) {};
\node[sng] (9) at ($(zero) + (0*\x,-1)$) {};

\coordinate (zero) at (\d-.5,0);

\draw[rounded corners,fill=lightgray!30,lightgray!30] ($ (zero) + (-.5,-\y-.4) $) rectangle ($(zero) + (1.1+ 3*\x, \y+.3)$);

\node[sng] (if) at (zero) {};
\node[sng] (1f) at ($(zero) + (1*\x,0)$) {};
\node[sng] (2f) at ($(zero) + (2*\x,0)$) {};
\node[sn] (3f) at ($(zero) + (3*\x,1)$) {};
\node[sn] (4f) at ($(zero) + (3*\x,0)$) {};
\node[sn] (5f) at ($(zero) + (3*\x,-1)$) {};
\node[sn] (6f) at ($(zero) + (2*\x,1)$) {};
\node[sn] (7f) at ($(zero) + (2*\x,-1)$) {};
\node[sng] (8f) at ($(zero) + (1*\x,-1)$) {};
\node[sng] (9f) at ($(zero) + (0*\x,-1)$) {};

\coordinate (zero) at (2*\d-.5,0);

\draw[rounded corners,fill=lightgray!30,lightgray!30] ($ (zero) + (-.5,-\y-.4) $) rectangle ($(zero) + (1.1+ 3*\x, \y+.3)$);

\node[sng] (iff) at (zero) {};
\node[sng] (1ff) at ($(zero) + (1*\x,0)$) {};
\node[sng] (2ff) at ($(zero) + (2*\x,0)$) {};
\node[sn] (3ff) at ($(zero) + (3*\x,1)$) {};
\node[sn] (4ff) at ($(zero) + (3*\x,0)$) {};
\node[sn] (5ff) at ($(zero) + (3*\x,-1)$) {};
\node[sn] (6ff) at ($(zero) + (2*\x,1)$) {};
\node[sn] (7ff) at ($(zero) + (2*\x,-1)$) {};
\node[sn] (8ff) at ($(zero) + (1*\x,-1)$) {};
\node[sn] (9ff) at ($(zero) + (0*\x,-1)$) {};

\path
(0,.55) edge (i)
(3.6,5.05) edge (io)

(io) edge node[above] {\footnotesize$a$} (1o)
(1o) edge[bend left=20] node[above] {\footnotesize$\#$} (2o)
(2o) edge[bend left=20] node[below] {\footnotesize$a$} (1o)
(2o) edge[bend right =10] node[below,pos=.6] {\footnotesize$b$} (3o)
(2o) edge[bend left =10] node[above,pos=.6] {\footnotesize$b$} (5o)
(3o) edge[bend left=0] node[right] {\footnotesize$\#$} (4o)
(5o) edge[bend right=0] node[right] {\footnotesize$\#$} (4o)
(4o) edge[loop right] node[right] {\footnotesize$\#$} (4o)
(3o) edge[bend left=20] node[below,pos=.8] {\footnotesize$\#$} (6o)
(6o) edge[bend left=20] node[above] {\footnotesize$b$} (3o)
(5o) edge[bend right=0] node[above,pos=.7] {\footnotesize$\#$} (7o)
(7o) edge node[above] {\footnotesize$b$} (8o)
(8o) edge node[above] {\footnotesize$\#$} (9o)

(i) edge node[above] {\footnotesize$a$} (1)
(1) edge[bend left=20] node[above] {\footnotesize$\#$} (2)
(2) edge[bend left=20] node[below] {\footnotesize$a$} (1)
(2) edge node[below] {\footnotesize$b$} (3)
(2) edge node[above] {\footnotesize$b$} (5)
(3) edge[bend left=0] node[below,pos=.2] {\footnotesize$\#$} (4f)
(5) edge[bend right=0] node[above,pos=.2] {\footnotesize$\#$} (4f)
(3) edge[bend left=0] node[above] {\footnotesize$\#$} (6f)
(5) edge[bend right=0] node[below] {\footnotesize$\#$} (7f)

(3f) edge node[right] {\footnotesize$\#$} (4f)
(5f) edge node[right] {\footnotesize$\#$} (4f)
(4f) edge[loop right] node[right] {\footnotesize$\#$} (4f)
(3f) edge[bend left=0] node[below] {\footnotesize$\#$} (6f)
(5f) edge node[above] {\footnotesize$\#$} (7f)
(7f) edge[bend right=20] node[below] {\footnotesize$c$} (8ff)

(3ff) edge node[right] {\footnotesize$\#$} (4ff)
(5ff) edge node[right] {\footnotesize$\#$} (4ff)
(4ff) edge[loop right] node[right] {\footnotesize$\#$} (4ff)
(3ff) edge[bend left=20] node[below] {\footnotesize$\#$} (6ff)
(6ff) edge[bend left=20] node[above] {\footnotesize$c$} (3ff)
(5ff) edge node[above] {\footnotesize$\#$} (7ff)
(7ff) edge node[above] {\footnotesize$c$} (8ff)
(8ff) edge node[above] {\footnotesize$\#$} (9ff)
(9ff.south) edge[bend right=20] node[below] {\footnotesize$c$} (5ff.south)

;

\path[draw,-stealth, rounded corners] (6f.north) -- ($(6f.north) + (0,.3)$) -- node[below] {\footnotesize$ c$}  ($(3ff.north) + (0,.3)$) -- (3ff.north);

\path[draw,-stealth, rounded corners] (9o.south) --  ($(9o.south) + (0,-.3)$) -- node[below] {\footnotesize$b$}  ($(5o.south) + (0,-.3)$) -- (5o.south);

    \end{tikzpicture}  
    }
    \caption{The construction of $\aut$' (bottom, for the sake of readability, we only sketch the letter processed by a transition, not the stack operations, and unreachable states are in gray) from $\aut$ (top). We assume that $C$ contains all configurations with a white state.
    A run proceeds from the left copy to the middle copy if the current configuration of $\aut$ is in $C$. In the middle and right copy, every $b$-label on a transition is replaced by a $c$. All transitions in the left and middle component are rejecting, while the transitions of the right copy have the same colors as in $\aut$.} 
    \label{fig:construction}
\end{figure}

To this end, let 
\[w = (a\#)^n (b\#)^n (c\#)^n \#^\omega \in L_{abc}\] and define 
\[w_1 = (a\#)^n (b\#)^n \#^\omega \qquad \text{ 
and } \qquad
w_2 = (a\#)^n (b\#)^{2n} \#^\omega, \]
which are both in $L$. 
Thus, let $\rho_1$ and $\rho_2$ be the accepting runs of $\aut$ on $w_1$ and $w_2$ induced by the resolver~$\hstrat$.
A prefix~$\rho_1'$ of $\rho_1$ processing $w_1' = (a\#)^n (b\#)^{n-1}b$ is also a prefix of $\rho_2$ processing $w_1'$ (note that we have removed the last $\#$, as the resolver inducing the runs has access to the next letter to be processed). 
As $\rho_1$ is an accepting run of $\aut$ on $w_1 = w_1'\#^\omega$, the last configuration of $\rho_1'$ is in $C$. 

Hence, $\aut'$ can simulate the run prefix~$\rho_1'$ processing $w_1'$, switch the first flag and then continue to simulate the suffix of $\rho_2'$ obtained by removing $\rho_1'$, which processes $\#(c\#)^n \#^\omega$. 
This run is accepting, i.e., we have $w = w_1'\#(c\#)^n \#^\omega \in L(\aut')$.

Now, let $w \in L(\aut')$, i.e., there is an accepting run~$\rho'$ of $\aut'$ on $w$. 
By construction of $\aut'$, we can split $\rho'$ into a finite prefix~$\rho_p'$ before the first flag is switched and the corresponding infinite suffix~$\rho_s'$ starting with the switch.
Again, by construction, $\rho_p'$ is the simulation of a run~$\rho_p$ of $\aut$ that processes the same input and ends in a configuration in $C$, and no prefix of $\rho_p$ ends in $C$. 
Hence, we can conclude that both $\rho_p$ and $\rho_p'$ process $(a\#)^n (b\#)^{n-1}b$ for some $n >0$, as these are the minimal words leading to a configuration from which $\#^\omega$ can be accepted.

Now, consider the suffix~$\rho_s'$, which processes at least one $c$.
It also simulates a run suffix~$\rho_s$ of $\aut$ and $\ell(\rho_s)$ is obtained from $\ell(\rho_s')$ by replacing each $c$ by a $b$.
Furthermore, $\rho_s$ starts in the last configuration of $\rho_p$ and satisfies the acceptance condition, as $\rho_s'$ satisfies the acceptance condition.
Hence, $\rho_s$ processes $\#(b\#)^n\#^\omega$, as the concatenation of $\rho_p$ and $\rho_s$ is an accepting run.
Altogether, $\rho_p'$ processes $(a\#)^n (b\#)^{n-1}b$ and $\rho_s'$ processes $\#(c\#)^n\#^\omega$, i.e., $w = (a\#)^n (b\#)^n (c\#)^n\#^\omega$, which is in $L_{abc}$.

We conclude that $L(\aut')=L_{abc}$, which contradicts $L_{abc} \notin$ \cfl. 
 Therefore, $L$ separates \gfgcfl and \cfl.
\end{proof}

The proof of Theorem~\ref{thm_vpl} is based on another language separating \gfgcfl and \cfl. 
A third such language is the following, based on palindromes. 
Let $h\colon \set{0,1,\#}^*\rightarrow \set{0,1}^* $ be the homomorphism induced by $h(0) = 0$, $h(1)=1$, and $h(\#)=\epsilon$.
Define 
\[P = \set{ v\#^\omega \mid h(v)  = x \rev{x} \text{ for some } x \in \set{0,1}^*},\]
where $\rev{x}$ denotes the reversal of $x$. The proof that $P$ indeed separates \cfl and \gfgcfl can be found in the appendix.

%% file: content/games.tex
In this section, we show that the winner of infinite-duration games with \gfgcfl winning conditions can be effectively determined. 
This result is best phrased in terms of Gale-Stewart games, abstract games without an arena~\cite{GaleStewart53}, as we are interested in the influence of the winning condition on the decidability of solving games.\footnote{Games in finite arenas can easily be encoded as Gale-Stewart games.} 

Formally, a Gale-Stewart game~$\gsgame(L)$ is given by an $\omega$-language~$L \subseteq (\SigmaI \times \SigmaO)^\omega$. 
It is played between Player~$1$ and Player~$2$ in rounds~$n = 0,1,2, \ldots$:
In each round, first Player~$1$ picks a letter~$a_1(n) \in \SigmaI$, then Player~$2$ picks a letter~$a_2(n) \in \SigmaO$.
After $\omega$ rounds, the players have constructed an outcome
\[w=
\binom{a_1(0)}{a_2(0)}
\binom{a_1(1)}{a_2(1)}
\binom{a_1(2)}{a_2(2)} \cdots \]
which is winning for Player~$2$ if it is in $ L$.
A strategy for Player~$2$ in $\gsgame(L)$ is a mapping~$\sigma \colon \SigmaI^+ \rightarrow \SigmaO$. 
The outcome~$w$ is consistent with $\sigma$, if $a_2(n) = \sigma(a_1(0) \cdots a_1(n))$ for all $n$. 
A strategy~$\sigma$ for Player~$2$ is winning if every outcome that is consistent with $\sigma$ is in $L$. 
Player~$2$ wins $\gsgame(L)$ if she has a winning strategy for $\gsgame(L)$.

\begin{propC}[\cite{DBLP:journals/tcs/Finkel01a,%
DBLP:journals/iandc/Walukiewicz01}]
\label{prop_games}
\hfill
\begin{enumerate}
    \item\label{prop_games_nondet} The following problem is undecidable: Given an $\omega$-PDA~$\aut$, does Player~$2$ win $\gsgame(L(\aut))$?
    
    \item\label{prop_games_det} The following problem is $\exptime$-complete: Given an $\omega$-DPDA~$\aut$, does Player~$2$ win $\gsgame(L(\aut))$?
    
\end{enumerate}
\end{propC}

Walukiewicz's decidability result~\cite{DBLP:journals/iandc/Walukiewicz01} is formulated for parity games on configuration graphs of pushdown automata.
However, a Gale-Stewart game with \dcfl winning condition can be reduced in polynomial time to a parity game on a configuration graph of a pushdown machine, and vice versa.
This construction crucially depends on the determinism of the automaton recognizing the winning condition, as witnessed by the undecidability result for winning conditions recognized by (possibly nondeterministic) $\omega$-PDA.

Our main result shows that decidability extends to games given by $\omega$-GFG-PDA, i.e., not all types of nondeterminism lead to undecidability.

\begin{thm}
\label{thm_gfggfg}
The following problem is $\exptime$-complete: Given an $\omega$-GFG-PDA~$\aut$, does Player~$2$ win $\gsgame(L(\aut))$?
\end{thm}

\begin{proof}
Given $\aut = (Q, \Sigma, \Gamma, q_\initmark, \Delta, \col)$ with $\Sigma = \SigmaI \times \SigmaO$ we construct an $\omega$-DPDA~$\autud$ such that Player~$2$ wins $\gsgame(L(\aut))$ if and only if she wins $\gsgame(L(\autud))$. 
This yields $\exptime$ membership, as determining the winner of Gale-Stewart games with \dcfl winning conditions is in $\exptime$ (see Proposition~\ref{prop_games}(\ref{prop_games_det})) and the size of $\autud$ is polynomial in the size of $\aut$. 
The matching lower bound is immediate, as determining the winner of games with \dcfl winning conditions is already $\exptime$-hard (again, see Proposition~\ref{prop_games}(\ref{prop_games_det})).

Intuitively, we construct an $\omega$-DPDA~$\autud$ that processes simultaneously both an input of $\aut$ and a run of $\aut$, and checks whether the run is indeed an accepting run of $\aut$ on the input. In the corresponding Gale-Stewart game with winning condition $L(\autud)$, Player~$2$ has to both choose a letter in $\SigmaO$ and transitions of $\aut$. In other words, we have moved the nondeterminism of $\aut$ into Player~$2$'s moves. Then a winning strategy for Player~$2$ in $\gsgame(L(\aut))$ can be combined with a resolver to form a winning strategy in $\gsgame(L(\autud))$ and vice versa.

Formally, we construct $\autud$ such that it recognizes all $\omega$-words~$m_0 m_1 m_2 \cdots$ over $\SigmaI  \times (\SigmaO \cup \Delta)$ where each \emph{block}~$m_j$ is of the form
\[
\binom{a_1(j)}{a_2(j)}
\binom{b_{j,0}}{\tau_{j,0}}
\cdots
\binom{b_{j,n_j-1}}{\tau_{j,n_j-1}}
\binom{b_{j,n_j}}{\tau_{j,n_j}}
\]
for some $n_j \ge 0$ satisfying the following conditions:
\begin{enumerate}
    
    \item The transitions~$\tau_{j,0}, \ldots, \tau_{j,n_j-1}$ are $\epsilon$-transitions, and $\tau_{j,n_j}$ is an $\binom{a_1(j)}{a_2(j)}$-transition.
    
    \item The sequence~$\tau_{0,0} \cdots \tau_{0,n_0} \tau_{1,0} \cdots \tau_{1,n_1}  \tau_{2,0} \cdots \tau_{2,n_2} \cdots $ of transitions induces an accepting run of $\aut$ on the $\omega$-word~$\binom{a_1(0)}{a_2(0)} \binom{a_1(1)}{a_2(1)} \binom{a_1(2)}{a_2(2)} \cdots $.
    Note that all the~$b_{j,j'}$ picked by Player~$1$ are ignored while Player~$2$ constructs the run, only the letters~$a_1(j)$ are relevant.

\end{enumerate}
If $w$ is of that form then the decomposition into blocks is unique.

An $\omega$-DPDA~$\autud$ recognizing this language can easily be constructed in polynomial time from $\aut$.
To this end, $\autud$ deterministically simulates the transitions given in a block on the letter from $\SigmaI \times \SigmaO$ at the beginning of the block.
If a transition is not applicable, then the run terminates and is therefore rejecting.
Some standard constructions are necessary to ensure that the input has the right format; in particular, we need to adapt the coloring to rule out that from some point onwards only $\epsilon$-transitions appear in the input.
It remains to show that Player~$2$ wins $\gsgame(L(\aut))$ if and only if she wins $\gsgame(L(\autud))$.

First, we construct a mapping~$\sigma \mapsto \sigma_d$ turning a winning strategy~$\sigma$ for Player~$2$ in $\gsgame(L(\aut))$ into a winning strategy~$\sigma_d$ for Player~$2$ in $\gsgame(L(\autud))$.
To this end, we fix a resolver~$\hstrat \colon \Delta^* \times \Sigma \rightarrow \Delta$ for $\aut$. 

Intuitively, the strategy~$\sigma_d$ alternates between simulating a move of $\sigma$ and then uses $\hstrat$ to construct a sequence of transitions that processes the letter determined by the move.
This sequence starts with a finite number of $\epsilon$-transitions followed by one transition processing the letter.

More formally, define $\sigma_d$ inductively starting with $\sigma_d(a) = \sigma(a)$ for $a \in \SigmaI$.
Now, let $v = a_1(0) \cdots a_1(n) \in \SigmaI^*$ with $n > 0$ be an input such that 
\[
a_2(j) = \sigma_d(a_1(0) \cdots a_1(j))
\]
is already defined for every $j <n$.
To define $\sigma_d(v)$ we consider two cases.

If $a_2(n-1) \in \SigmaO \cup \Delta $ is a non-$\epsilon$-transition, then we define~$\sigma_d(v) = \sigma( v') $, where $v'$ is obtained from $v$ by removing the letters at positions~$j$ with $a_2(j) \in \Delta$.
This simulates the next move of $\sigma$, as the transition~$a_2(n-1)$ has processed the last letter.
In the other case (i.e., if $a_2(n-1)$ is either a letter in $\SigmaO$ or an $\epsilon$-transition) define $\sigma_d(v) = \hstrat( \rho, \binom{a_1(j')}{a_2(j')} ) $ where
$\rho \in \Delta^*$ is obtained from $a_2(0) \cdots a_2(n-1)$ by removing the letters at positions~$j$ with $a_2(j) \in \SigmaO$ and where $j' < n$ is maximal with $a_2(j') \in \SigmaO$. 
This move continues the construction of a run infix that processes the last letter from $\SigmaI \times \SigmaO$, which appears at position~$j'$.

Now, let $w' \in (\SigmaI \times (\SigmaO \cup \Delta))^\omega$ be consistent with $\sigma_d$.
An induction shows that $w'$ is a sequence of blocks that encodes an outcome~$w$ over $\SigmaI \times \SigmaO$ that is consistent with $\sigma$, and a run~$\rho$ of $\aut$ on $w$ induced by $\hstrat$. 
As $\sigma$ is a winning strategy, $w$ is in $L(\aut)$, which implies that $\rho$ is accepting. 
Hence, $w'$ is in $L(\autud)$, i.e., $\sigma_d$ is indeed winning for Player~$2$ in $\gsgame(L(\autud))$.

Conversely, we construct a mapping~$\sigma \mapsto \sigma_{-d}$ turning a winning strategy~$\sigma$ for Player~$2$ in $\gsgame(L(\autud))$ into a winning strategy~$\sigma_{-d}$ for Player~$2$ in $\gsgame(L(\aut))$.
Intuitively, to define $\sigma_{-d}$, we simulate a play in $\gsgame(L(\aut))$ by a play in $\gsgame(L(\autud))$, and copy the choice of letters made by $\sigma$ while ignoring the moves building the run of $\aut$.

Fix some $\dummy \in \SigmaI$\label{page:dummyletter}, which we use as dummy input.
We inductively define for every input~$v \in \SigmaI^+$ for $\sigma_{-d}$ an input~$v_d \in \SigmaI^+$ for $\sigma$ and then define $\sigma_{-d}(v) = \sigma(v_d)$.
We begin by defining $a_d = a$ for every $a \in \SigmaI$.
Now, assume we have defined $v_d$ for some $v$. 
Let $n$ be minimal such that $\sigma(v_d \dummy^n) \in \SigmaO \cup \Delta$ is a non-$\epsilon$-transition. 
Then, we define $(va)_d = v_d\dummy^na$.
Intuitively, we extend $v_d$ by irrelevant inputs until $\sigma$ completes the run infix processing the last letter. 

Let $w \in (\SigmaI \times \SigmaO)^\omega$ be consistent with $\sigma_{-d}$. 
An induction shows that there is an $w' \in (\SigmaI \times (\SigmaO \cup \Delta))^\omega$ that is consistent with $\sigma$ that encodes $w$ and an accepting run of $\aut$ on $w$. 
Hence, $w \in L(\aut)$,  i.e., $\sigma_{-d}$ is indeed winning for Player~$2$ in $\gsgame(L(\aut))$.
\end{proof}

Recall that a game is determined if either of the players has a winning strategy.
Translating a winning strategy for Player~$1$ (defined as expected) for $\gsgame(L(\autud))$ into a winning strategy for him in $\gsgame(L(\aut))$ implies that $\gsgame(L(\aut))$ is determined, as $\gsgame(L(\autud))$ is determined.
Such a transformation can be obtained along the lines of the transformations presented above for Player~$2$ winning strategies.

As universality of $L \subseteq \Sigma^\omega$ is equivalent to Player~$2$ winning $\gsgame(\set{\binom{w}{\#^\omega} \mid w \in L})$ and as $\set{\binom{w}{\#^\omega} \mid w \in L}$ is in \gfgcfl if $L$ is in \gfgcfl, we obtain the following corollary of our main theorem. 

\begin{cor}
\label{corollary_universality}
The following problem is in $\exptime$: Given an $\omega$-GFG-PDA~$\aut$, is $L(\aut)$ universal?
\end{cor}

This contrasts with the universality problem for $\omega$-PDA, which is undecidable.
Emptiness of $\omega$-GFG-PDA is also decidable, as it is decidable for $\omega$-PDA, while equivalence of $\omega$-GFG-PDA is undecidable, as it is already undecidable for $\omega$-DPDA with Büchi or co-Büchi acceptance conditions~\cite{equivalenceBuchi}.

Also, let us mention that one can apply the reduction presented in the proof of Theorem~\ref{thm_gfggfg} also if $\aut$ is not known to be good-for-games.
If Player~$2$ wins $\gsgame(L(\autud))$, then she wins $\gsgame(L(\aut))$ as well. 
However, if Player~$2$ does not win $\gsgame(L(\autud))$, then she might or might not win $\gsgame(L(\aut))$, i.e., the reduction is sound, but not complete, if $\aut$ is not good-for-games.
The same holds true for Corollary~\ref{corollary_universality}.

While we only consider the realizability problem here, i.e., the problem of determining whether Player~$2$ wins the game, our proof of Theorem~\ref{thm_gfggfg} can be extended to the synthesis problem, i.e., the problem of computing a winning strategy for Player~$2$, if she wins the game.
Such a strategy can be finitely represented by a deterministic pushdown automaton with output (called pushdown transducers, or PDT) reading finite sequences over $\SigmaI$ and outputting a single letter from~$\SigmaO$. 
These are efficiently computable for Gale-Stewart games with \dcfl winning conditions~\cite{DBLP:journals/corr/abs-1006-1415,DBLP:journals/iandc/Walukiewicz01}.
Hence, one can compute a winning strategy for Player~$2$ in $\gsgame(L(\autud))$ and then apply the transformation described in the second part of the proof of Theorem~\ref{thm_gfggfg}, which is implementable by deterministic pushdown transducers. 

\begin{thm}
\label{thm:transducer}
Let $\aut$ be an $\omega$-GFG-PDA. 
If Player~$2$ wins $\gsgame(L(\aut))$, then she has a winning strategy that is implemented by a PDT. Furthermore, such a PDT can be computed in exponential time in $\size{\aut}$.
\end{thm}

The details of this straightforward, but slightly tedious, construction are presented in the appendix.

On the other hand, there are games with \gfgcfl winning condition that are won by Player~$1$, but not with a winning strategy implementable by pushdown transducers.
Hence, the situation between the players is asymmetric.
Again, the example proving this claim is presented in the appendix.

%% file: content/closure.tex
In Section~\ref{sec_gfg}, we have shown that \gfgcfl is a new subclass of $\omega$-contextfree languages.
Here, we study the closure properties of this class, which differ considerably from those of \dcfl and \cfl.

We say that a class~$\mathcal{L}$ of $\omega$-languages is closed under ($\epsilon$-free) homomorphisms if $\set{ f(w)  \mid w \in L } $ is in $ \mathcal{L}$
for every $L \in \mathcal{L}$ and every $f \colon \Sigma \rightarrow (\Sigma')^+$.
Here, we disallow $\epsilon$ in the image of $f$ to ensure that $f(w)$ is infinite.

\begin{thm}
\label{thm_closure}
\gfgcfl is not closed under union, intersection, complementation, set difference, nor homomorphism. 
\end{thm}

\begin{proof}
\textbf{Union:} As \dcfl contains both $L_1 = \set{ (a\#)^n (b\#)^n \#^\omega \mid n\ge 1}$ and $L_2 = \set{ (a\#)^n (b\#)^{2n} \#^\omega \mid n\ge 1}$ whose union is not in \gfgcfl (Theorem~\ref{thm_gfgvscfl}), Proposition~\ref{prop_inclusions} implies that \gfgcfl is not closed under union. 

\textbf{Intersection:} As \dcfl contains $L_1 = \set{a^n b^n a^* b^\omega \mid n\ge 1}$ and $L_2 = \set{a^* b^n a^n b^\omega \mid n\ge 1}$ whose intersection~$L_1\cap L_2 = \set{a^n b^n a^n b^\omega \mid n\ge 1}$ is not even in \cfl~\cite[Proposition~1.3]{DBLP:journals/jcss/CohenG77a}, Proposition~\ref{prop_inclusions} implies that \gfgcfl is not closed under intersection.

\textbf{Complementation:} We show that the complement of the language~$\seplang \in $ \gfgcfl used in Section~\ref{sec_gfg} to separate \dcfl and \gfgcfl is not even in \cfl.
A word~$w$ is not in $\seplang $ if $\proj_i(w)$ does not have a safe suffix for both $i \in \set{1,2}$, which due to Remark~\ref{remark_safe} is equivalent to 
\[\liminf_{n \rightarrow \infty} \el(\proj_i(w(0) \cdots w(n))) = -\infty\] for both $i$. 

Towards a contradiction, assume there is an $\omega$-PDA~$\aut = (Q, \Sigma, \Gamma, q_\initmark, \Delta, \col)$ with $L(\aut) = \Sigma^\omega \setminus \seplang $. 
As in the proof of Lemma~\ref{lemma_Lisnotdcfl}, define $x_1 = \binom{\plus}{\zero} \binom{\plus}{\minus}$ and $x_2 = \binom{\zero}{\plus} \binom{\minus}{\plus}$.
Recall that every infix of length at least $3$ of a word built by concatenating copies of the $x_i$ has a strictly positive energy level in one component.

Again, we define
\[w_{\overline{ss}} =
x_1 \, 
(x_2)^3 \,
(x_1)^7 \,
(x_2)^{15} \,
(x_1)^{31} \,
(x_2)^{63} \,
\cdots
\]
which satisfies \[\el( \proj_1(  x_1 (x_2)^3 \cdots (x_2)^{2^{2j}-1} ) ) = -j \] for every $j > 1$ and \[\el( \proj_2(  x_1 (x_2)^3 \cdots (x_1)^{2^{2j-1}-1} ) ) = -j \] for every $j > 0$ (see Page~\pageref{bla}). 
Hence, $w_{\overline{ss}} \in \Sigma^\omega \setminus \seplang $, i.e., there is an accepting run~$\rho = c_0 \tau_0 c_1 \tau_1 c_2 \tau_2 \cdots$ of $\aut$ on $w$. 

Recall that a step of $\rho$ is a position~$n$ such that $\sh(c_n) \le \sh(c_{n+j})$ for all $j \ge 0$.
Every infinite run has infinitely many steps.
Hence, we can find two steps~$s < s'$ satisfying the following properties:

\begin{enumerate}
    
    \item There is a state~$q \in Q$ and a stack symbol~$X \in \Gammabot$ such that $c_s = (q, \gamma X)$ and $c_{s'} = (q, \gamma' X)$ for some $\gamma, \gamma'$, i.e., both configurations have the same state and topmost stack symbol.
    
    \item The maximal color labeling the sequence~$\tau_{s} \cdots \tau_{s'-1}$ of transitions leading from $c_s$ to $c_{s'}$ is even.

    \item The sequence~$\tau_{s} \cdots \tau_{s'-1}$  processes an infix~$v$ of $w$ with $\el(\proj_i(v)) > 0$, for some $i \in \set{1,2}$ (recall that every sufficiently long infix of a word constructed from the $x_i$ has positive energy level in one component).
    
\end{enumerate}

Consider the sequence~$\tau_0 \cdots \tau_{s-1}( \tau_s \cdots \tau_{s'-1} )^\omega $ of transitions.
Due to the first property, it induces a run~$\rho'$ of $\aut$, which is accepting due to the second property. 
Finally, due to the third property, $\rho'$ processes a word with suffix~$v^\omega$.
Such a word has a safe suffix in component~$i$, as $\el(\proj_i(v)) > 0$.

Hence, we have constructed a word~$w$ in $\seplang $ such that there is an accepting run of $\aut$ on $w$, i.e., we have derived a contradiction to $L(\aut) = \Sigma^\omega \setminus \seplang $.
As we have picked $\aut$ arbitrarily, we have shown that $\Sigma^\omega \setminus \seplang $ is not in \cfl.
Thus, due to Proposition~\ref{prop_inclusions} and Lemma~\ref{lemma_Lisgfg}, \gfgcfl is not closed under complementation.

\textbf{Set difference:} As $\Sigma^\omega$ is in \dcfl $\subseteq$ \gfgcfl for every alphabet~$\Sigma$, \gfgcfl cannot be closed under set difference, as complementation is set difference with $\Sigma^\omega$. 

\textbf{Homomorphism:} As \dcfl contains 
\begin{align*}
L ={} &{}
\left\{\,
\left[\,\binom{a}{1}\binom{\#}{1} \,\right]^n 
\left[\,\binom{b}{1}\binom{\#}{1}\,\right]^{n} 
\binom{\#}{1}^\omega \,\middle|\, n\ge 1 \right\} \cup\\
&{}\left\{ \,
\left[\,\binom{a}{2}\binom{\#}{2}\,\right]^n 
\left[\,\binom{b}{2}\binom{\#}{2}\,\right]^{2n} 
\binom{\#}{2}^\omega \,\middle|\, n\ge 1 \right\}
\end{align*}
whose projection (which is a homomorphism)
\[\proj_1(L) = 
\set{ (a\#)^n (b\#)^n \#^\omega \mid n\ge 1} \cup \set{ (a\#)^n (b\#)^{2n} \#^\omega \mid n\ge 1}
\]
is not in \gfgcfl (Theorem~\ref{thm_gfgvscfl}), Proposition~\ref{prop_inclusions} implies that \gfgcfl is not closed under homomorphisms. 
\end{proof}

As we only used languages in \dcfl $\subseteq$ \gfgcfl to witness the failure of closure under intersection, union, and set difference, we obtain the following corollary. 

\begin{cor}
\gfgcfl is not closed under union, intersection, and set difference with languages in \dcfl.
\end{cor}

Finally, using standard arguments one can show that closure under these operations with $\omega$-regular languages holds, as it does for \dcfl and \cfl.

\begin{thm}
\label{theorem_closurereg}
If $L \in $ \gfgcfl and $R$ is $\omega$-regular, then $L \cap R$, $L \cup R$, and $L \setminus R$ are in \gfgcfl as well.
\end{thm}

\begin{proof}
Let $L = L(\aut)$ for some $\omega$-GFG-PDA~$\aut$ and $R$ be $\omega$-reg\-u\-lar, i.e., $R = L(\auta)$ for some deterministic parity automaton~$\auta$ (see, e.g., \cite{GTW02} for definitions). 
Furthermore, let $\aut \times \auta $ be the product automaton of these two automata, which is again an $\omega$-PDA that simulates a run of $\aut$ and the unique run of $\auta$  simultaneously. 

Using a detour via the Muller acceptance condition and the LAR construction (see for example~\cite{GTW02}), one can turn $\aut \times \auta$ into automata $(\aut \times \auta)_\cap$, $(\aut \times \auta)_\cup$, and $(\aut \times \auta)_\backslash$ such that the following holds true:
\begin{itemize}
 
    \item A run of $(\aut \times \auta)_\cap$ is accepting if the simulated run of $\aut$ and the simulated run of $\auta$ are accepting.
  
   \item A run of $(\aut \times \auta)_\cup$ is accepting if either the simulated run of $\aut$ or the simulated run of $\auta$ is accepting.
    
    \item A run of $(\aut \times \auta)_\backslash$ is accepting if the simulated run of $\aut$ is accepting and the simulated run of $\auta$ is not accepting.
    
\end{itemize}
All three automata can be shown to be good-for-games, as only the nondeterminism of $\aut$ has to be resolved:
A resolver for $\aut$ can easily be turned into one for the three automata that just ignores the additional inputs stemming from taking the product (note that this crucially depends on $\auta$ and the LAR memory being deterministic).
\end{proof}

Note that $R \setminus L$ is not necessarily in \gfgcfl (not even in \cfl) if $R$ is $\omega$-regular and $L$ is in \gfgcfl.

%% file: content/decisions.tex
In this section, we show that deciding good-for-gameness is, unfortunately, undecidable, both for automata and for languages.
These results contrast with good-for-gameness being decidable for parity automata~\cite{KS15}, even in polynomial time for B\"uchi and co-B\"uchi automata~\cite{KS15,BK18}, and every $\omega$-regular language being good-for-games, as deterministic parity automata recognize all $\omega$-regular languages.

To show these undecidability results, we introduce some additional notation. 
If $\aut$ is an $\omega$-PDA and $q$ one of its states, then we write $L(\aut,q)$ for the language accepted by the automaton obtained from $\aut$ by replacing its initial state with~$q$. 

\begin{thm}
The following problems are undecidable: 
\begin{enumerate}
    \item\label{thm_undec_1} Given an $\omega$-PDA~$\aut$, is $\aut$ good-for-games?
    \item\label{thm_undec_2} Given an $\omega$-PDA~$\aut$, is $L(\aut) \in $ \gfgcfl?
\end{enumerate}

\end{thm}

\begin{proof}
Both proofs proceed by reduction from an undecidable problem for PDA over finite words (see, e.g., \cite{DBLP:books/daglib/0016921}). 
Such an automaton has the same structure as an $\omega$-PDA, but the coloring~$\col$ is replaced by a set~$F$ of accepting states.
A finite run is accepting, if it ends in an accepting state.

We consider the following two problems:
\begin{itemize}
    \item DPDA inclusion: Given DPDA's $\daut_1$ and $\daut_2$, is $L(\daut_1) \subseteq L(\daut_2)$?
    
    \item PDA universality: Given a PDA~$\aut$ over $\Sigma$, is $L(\aut) = \Sigma^*$?
\end{itemize}

\ref{thm_undec_1}.) We reduce the inclusion problem for DPDA over finite words to deciding whether an $\omega$-PDA is good-for-ga\-mes. 
Since the inclusion problem is undecidable~\cite{Val73}, so is deciding whether an $\omega$-PDA is good-for-games.

Given DPDA~$\daut_1$ and $\daut_2$ over $\Sigma^*$, we first define infinitary versions of $\daut_1$ and $\daut_2$: Let $\aut_1$ and $\aut_2$ be $\omega$-DPDA over $(\Sigma\cup \set{\#})^\omega$, where $\#\notin \Sigma$, that are identical to $\daut_1$ and $\daut_2$ respectively, except with additional $\#$-transitions from states that are accepting in $\daut_1$ and $\daut_2$ to accepting sinks; all other states are made rejecting. Then $L(\aut_i)$ consists of words of the form $v\#w$ where $v\in L(\daut_i)$ and $w\in (\Sigma\cup \set{\#})^\omega$. 
We have $L(\aut_1) \subseteq L(\aut_2)$ if and only if $L(\daut_1) \subseteq L(\daut_2)$.

Consider $\aut$, an $\omega$-PDA over $\Sigma\cup \set{\#,\$}$ built as follows (see Figure~\ref{fig:intersectionreduction}): A fresh initial state~$q_I$ has $\$$-transitions to fresh states $q_1$ and $q_2$; from $q_1$ there is an $\epsilon$-transition to the initial state of $\aut_1$, and from $q_2$ there is a $\$$-transition to an accepting sink~$q_s$ and an $\epsilon$-transition to the initial state of $\aut_2$. None of these transitions manipulate the stack.  
We show that $\aut$ is good-for-games if and only if $L(\daut_1)\subseteq L(\daut_2)$.

\begin{figure}
\centering
    \begin{tikzpicture}[thick]

\draw[rounded corners,fill=lightgray] (-6,.8) rectangle (-3.1,-.8);
\draw[rounded corners,fill=lightgray] (6,-.8) rectangle (3.1,.8);

\node[state,minimum size =22] (i) at (0,0) {$q_\initmark$};
\node[state,minimum size =22] (s) at (2,-1.5) {$q_s$};
\node[state,minimum size =22] (q1) at (-2,0) {$q_1$};
\node[state,minimum size =22] (q2) at (2,0) {$q_2$};

\node[state,minimum size =22,fill=white] (i1) at (-4,0) {$q_\initmark^1$};
\node[state,minimum size =22,fill=white] (i2) at (4,0) {$q_\initmark^2$};

\path
(0,.8) edge (i)
(i) edge node[above] {$\$$} (q1)
(i) edge node[above] {$\$$} (q2)
(q1) edge node[above,near start] {$\epsilon$} (i1)
(q2) edge node[above,near start] {$\epsilon$} (i2)
(q2) edge node[right] {$\$ $} (s)
(s) edge[loop left] node[left]{$\Sigma\cup\set{\#,\$}$} ();

\node at (-5.5,0) {$\aut_1$};
\node at (5.5,0) {$\aut_2$};
    \end{tikzpicture}  
    \caption{The construction of $\aut$ from $\aut_1$ and $\aut_2$. All new transitions do not manipulate the stack, i.e., if $\aut_i$ is entered via the $\epsilon$-transition, then the initial configuration of $\aut_i$ is reached.}
    \label{fig:intersectionreduction}
\end{figure}
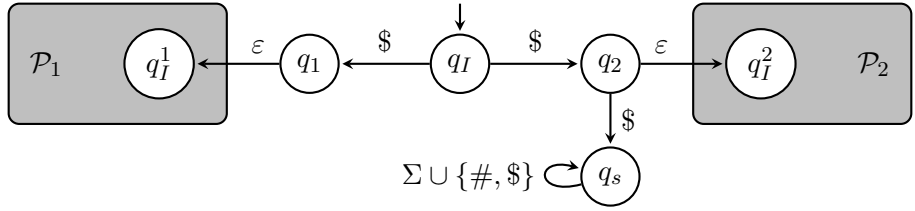

If $L(\daut_1)\subseteq L(\daut_2)$, then $L(\aut_1)\subseteq L(\aut_2)$ and $L(\aut,q_1)\subseteq L(\aut,q_2)$. Let $\hstrat$ be defined such that 
\begin{itemize}
    \item $\hstrat(\epsilon,\$)= (q_\initmark, \bot, \$, q_2, \bot)$,
    \item $\hstrat((q_\initmark, \bot, \$, q_2, \bot),\$)= (q_2, \bot, \$, q_s, \bot)$, and
    \item $\hstrat((q_\initmark, \bot, \$, q_2, \bot), a)= (q_2, \bot, \$, q_\initmark^2, \bot)$ where $a \neq \$$ and $q_\initmark^2$ denotes the initial state of $\aut_2$,
\end{itemize}
i.e., $\hstrat$ produces a run prefix that reaches either the accepting sink~$q_s$ or the initial state~$q_\initmark^2$ of $\aut_2$, which is deterministic. 
Hence, in both cases, there is no further nondeterminism to resolve.

We claim that $\hstrat$ is a resolver.
Let $w \in L(\aut)$, i.e., either $w = \$w'$ with 
\[
w' \in L(\aut, q_1) \cup L(\aut, q_2) \subseteq L(\aut, q_2),
\]
i.e., the second letter of $w$ is not equal to $\$$,
or $w = \$\$w'$ with $w' \in (\Sigma \cup \set{\#,\$})^\omega$.
In both cases, the run of $\aut$ on $w$ induced by $\hstrat$ is accepting. 

Conversely, if there is a word $v\in L(\daut_1)\setminus L(\daut_2)$, there is no resolver for $\aut$. 
Indeed, if a resolver $\hstrat$ chooses the transition~$(q_\initmark, \bot, \$, q_1,\bot)$ to $q_1$ to begin with, it is not a resolver since $\$^\omega\in L(\aut)$, but $\$^\omega \notin L(\aut,q_1)$. 
Finally, if $\hstrat$ chooses the transition~$(q_\initmark, \bot, \$, q_2,\bot)$ to $q_2$ to begin with, then $\hstrat$ is not a resolver either since $\$v\#^\omega \in L(\aut)$, but $v\#^\omega\notin L(\aut,q_2)$.

\ref{thm_undec_2}.) We proceed by a reduction from the universality problem for PDA over finite words, which is undecidable~\cite{DBLP:books/daglib/0016921}.
Namely, given a PDA~$\faut$ over $\Sigma$ we build an $\omega$-PDA~$\aut$ over $\Sigmahash \times \set{a,b,\#}$ with
$\Sigmahash =\Sigma \cup \set{ \# }$ such that $L(\aut)\in $ \gfgcfl if and only if $L(\faut)$ is universal.

First, note that 
\[ L_1 = \set{v\# w | v \in L(\faut) \text{ and } w\in \Sigmahash^\omega}\cup \Sigma^\omega\]
is universal if and only if $L(\faut)$ is universal.
 Furthermore, $\faut$ can easily be turned into an $\omega$-PDA recognizing $L_1$. 
 So, we can construct from $\faut$ an $\omega$-PDA~$\aut$ recognizing the language
 \[
L(\aut) =  \set{w \in (\Sigmahash\times  \set{a,b,\#})^\omega\mid \proj_1(w) \in L_1 \text{ or } \proj_2(w) \in L_2 },
\]
where $L_2$ is the $\omega$-language from the proof of Theorem~\ref{thm_gfgvscfl}, which is not in \gfgcfl but in \cfl.
We claim that $\aut$ has the desired property.
Trivially, if $L(\faut)$ is universal, then so is $L(\aut)$, which implies that $L(\aut)$ is in \gfgcfl.

Now assume that $L(\faut)$ is not universal which is witnessed by some word $v\#^\omega \notin L_1$, i.e, such that $v\notin L(\faut)$.
Towards a contradiction, assume that there is an $\omega$-GFG-PDA~$\raut$ with $L(\raut)=L(\aut)$, say with resolver~$\hstrat$.
We turn $\raut$ into another $\omega$-GFG-PDA~$\raut'$ recognizing the language
\[
\left\{ w \in \set{a,b,\#}^\omega \,\middle|\, \binom{v\#^\omega}{w} \in L(\raut) \right\},
\]
which yields the desired contradiction, as this language is equal to $L_2$, which is not in \gfgcfl.

To this end, we equip $\raut$ with a counter ranging over the positions of $v\#$ to simulate a run of $\raut$ on input~$\binom{v\#^\omega}{w}$ when given the input~$w$. 
As the counter behaves deterministically, no new nondeterminism is introduced when constructing $\raut'$ from $\raut$.
Hence, $\hstrat$ can be turned into a resolver for $\raut'$, which is therefore good-for-games, a contradiction.
\end{proof}

%% file: content/resolver.tex
As defined Section~\ref{sec_gfg}, a resolver is an arbitrary function, potentially even an uncomputable one, mapping a run prefix~$\rho$ and the next letter to be processed to a transition that is enabled in the last configuration of $\rho$.
Here, we study resource-bounded resolvers. 

For $\omega$-regular GFG automata, finite-state resolvers are always sufficient~\cite{HP06}.
Unsurprisingly, finite-state resolvers for $\omega$-GFG-PDA are too weak. In fact, we show that $\omega$-GFG-PDA with finite-state resolvers are determinizable.

On the other hand, it is tempting to conjecture that every $\omega$-GFG-PDA has a resolver that can be implemented by a pushdown automaton with output. 
However, we disprove this claim, i.e., even stronger devices are necessary in general.

Before we present our results, let us discuss one technical detail. 
Consider an $\omega$-GFG-PDA~$\aut$ and a resolver~$\hstrat$ for $\aut$.
As $\hstrat$ always has to return an enabled transition, which depends on the current state and the current top stack symbol of $\aut$, a machine computing $\hstrat$ would need to keep track of this information. 
However, this requires, in general, keeping track of the full stack of $\aut$.
This is beyond the capabilities of finite-state machines and blocks the stack of a pushdown machine implementing $\hstrat$ from performing other computations.
In order to better capture the complexity of resolving nondeterminism without involving the complexity of tracking enabled transitions, we equip the machine computing $\hstrat$ with an oracle providing the current top stack symbol of $\aut$.

We begin by considering finite-state resolvers. To this end, fix an $\omega$-GFG-PDA~$\aut = (Q, \Sigma, \Gamma, q_\initmark, \Delta, \col)$. A Moore machine~$\autm = (M, m_\initmark, \delta, \lambda)$ for $\aut$ consists of a finite set~$M$ of states with an initial state~$m_\initmark \in M$, a transition function~$\delta \colon M \times \Delta \rightarrow M$, and an output function~$\lambda \colon M \times \Sigma \times \Gamma_\bot \rightarrow \Delta $.
As usual, we extend the transition function to $\delta \colon \Delta^* \rightarrow M$ by defining $\delta(\epsilon) = m_\initmark$ and $\delta(\tau_0 \cdots \tau_{n-1}\tau_n) = \delta(\delta(\tau_0 \cdots \tau_{n-1}), \tau_n)$.

Let $\tau_0 \cdots \tau_n \in \Delta^*$ and $a \in \Sigma$. If the sequence~$\tau_0 \cdots \tau_n$ induces a run prefix of $\aut$, then let $X$ be the top stack symbol of its last configuration.
In this case, we define $f_\autm(\tau_0 \cdots \tau_n, a) = \lambda( \delta(\tau_0 \cdots \tau_n), a, X)$. 
If $\tau_0 \cdots \tau_n$ does not induce a run prefix, then we define $f_\autm(\tau_0 \cdots \tau_n, a)$ arbitrarily. 
We say that a resolver~$\hstrat$ for $\aut$ is a finite-state resolver, if there is some Moore machine~$\autm$ for $\aut$ such that $\hstrat = f_\autm$.
Note that the Moore machine $\autm$ on input~$\tau_0 \cdots \tau_n$ has to output a transition that is enabled in the last configuration of the run prefix induced by $\tau_0 \cdots \tau_n$.

\begin{exa}
Recall the $\omega$-GFG-PDA~$\aut$ in Example~\ref{example:pda} on Page~\pageref{example:pda} recognizing the language
\[
\set{a c^nd^n \#^\omega \mid n \ge 1} \cup \set{b c^nd^{2n} \#^\omega \mid n \ge 1}.
\]
It has a finite-state resolver, which is depicted in Figure~\ref{fig:fsresolver}.
Intuitively, it remembers the first letter of the input word processed, which suffices to resolve the nondeterminism at state~$q_1$ of $\aut$. 
This is the only nondeterministic choice to make, the rest of the run can be determined by keeping track of the current state of $\aut$ and using the information on the stack.
\begin{figure}
    \centering
    \begin{tikzpicture}[thick]
\def\y{2.75}
\def\x{2}
\def\e{1.5}
\tikzset{every state/.style = {minimum size =22}}
\node[state,fill=lightgray] (i) at (1*\x,0) {};
\node[state,fill=lightgray] (ua) at (2*\x, \y) {};
\node[state,fill=lightgray] (ub) at (2*\x, -\y) {};
\node[state,fill=lightgray] (ad) at (4.5*\x,\y) {};
\node[state,fill=lightgray] (bd1) at (4.5*\x,-\y) {};
\node[state,fill=lightgray] (bd2) at (7*\x,-\y) {};
\node[state,fill=lightgray] (acc) at (7*\x,\y) {};

\node[align=left,anchor=west,fill=lightgray!30] (li) at (1*\x+.5*\e,0) {$a,\bot \rightarrow (q_0, \bot, a, q_1, \bot A)$\\
$b,\bot \rightarrow (q_0, \bot, b, q_1, \bot B)$
};
\node[align=left,fill=lightgray!30] (lua) at (2*\x, \y+\e) {$c,X \rightarrow (q_1, X, c, q_1, XN)$\\ $d,N \rightarrow (q_1, N, d, q_2, \epsilon)$};
\node[align=left,fill=lightgray!30] (lub) at (2*\x, -\y-\e) {$c,X \rightarrow (q_1, X, c, q_1, XN)$\\$d,N\rightarrow (q_1, N, d, q_3, N)$};
\node[align=left,fill=lightgray!30] (lad) at (4.5*\x,\y+\e) {$d,N\rightarrow (q_2, N, d, q_2, \epsilon )$\\$\#,A \rightarrow (q_2, A, \#, q_4, \epsilon)$};
\node[align=left,fill=lightgray!30] (lbd1) at (4.5*\x,-\y-\e) {$d,N\rightarrow (q_3, N, d, q_5, \epsilon)$};
\node[align=left,fill=lightgray!30] (lbd2) at (7*\x,-\y-\e) {$d,N \rightarrow (q_5, N, d, q_3, N)$\\ $\#,B \rightarrow (q_5, B,\#,q_4, \epsilon)$};
\node[align=left,fill=lightgray!30] (lacc) at (7*\x,\y+\e) {$\#,\bot \rightarrow (q_4, \bot, \#, q_4, \bot )$};

 \path[-stealth]
 (1.25,0) edge (i)
 (i) edge[bend left] node[near start,anchor=west] {$(q_0, \bot, a, q_1, \bot A)$} (ua)
 (i) edge[bend right] node[near start,anchor=west] {$(q_0, \bot, b, q_1, \bot B)$} (ub)
 (ua) edge[loop below] node[right] {$(q_1, X, c, q_1, XN)$} ()
 (ub) edge[loop above] node[right] {$(q_1, X, c, q_1, XN)$} ()
 (ua) edge[] node[above] {$(q_1, N, d, q_2, \epsilon)$} (ad)
 (ub) edge[] node[below] {$(q_1, N, d, q_3, N)$} (bd1)
 (ad) edge[loop below] node[below] {$(q_2, N, d, q_2, \epsilon )$} ()
 (bd1) edge[bend left=10] node[above] {$(q_3, N, d, q_5, \epsilon)$} (bd2)
 (bd2) edge[bend left=10] node[below] {$(q_5, N, d, q_3, N)$} (bd1)
 (bd2) edge node[left] {$(q_5, B,\#,q_4, \epsilon)$} (acc)
 (ad) edge node[above] {$(q_2, A, \#, q_4, \epsilon)$} (acc)
 (acc) edge[out=240,in=210,looseness=8] node[left] {$(q_4, \bot, \#, q_4, \bot )$} (acc)
 ;
    \end{tikzpicture}
    \caption{A finite-state resolver for the $\omega$-GFG-PDA~$\aut$ from Example~\ref{example:pda}. Here, $X$ is again an arbitrary stack symbol of $\aut$ (not including $\bot$). All missing transitions lead to a rejecting sink (not depicted) and only relevant outputs are given (in light gray boxes near states), all others can be picked arbitrarily.}
    \label{fig:fsresolver}
\end{figure}
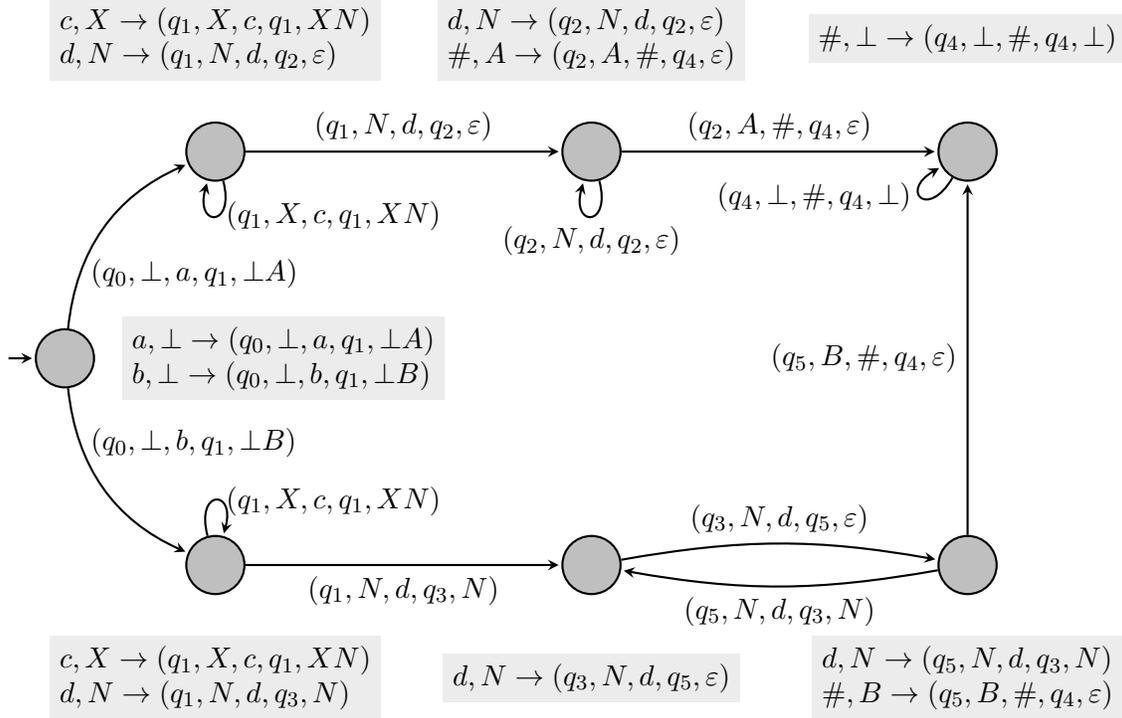
\end{exa}

Recall that the language recognized by the $\omega$-GFG-PDA in Example~\ref{example:pda} is in \dcfl. 
This is no coincidence.

\begin{thm}
If an $\omega$-GFG-PDA~$\aut$ has a finite-state resolver, then $L(\aut) \in $ \dcfl.
\end{thm}

\begin{proof}
Given an $\omega$-GFG-PDA~$\aut = (Q, \Sigma, \Gamma, q_\initmark, \Delta, \col)$ and a Moore machine~$\autm = (M, m_\initmark, \delta, \lambda)$ for $\aut$ implementing a resolver, we construct an $\omega$-DPDA~$\daut$ with $L(\daut) = L(\aut)$. 
The set of states of $\daut$ is $(Q \times M) \cup (Q \times M \times \Sigma)$ with initial state~$(q_\initmark, m_\initmark)$, and its stack alphabet is $\Gamma$. 
Intuitively, $\daut$ simulates the run constructed by $\hstrat_\autm$. To this end, it has to keep track of a state of $\aut$ and a state of $\autm$. 
In state~$(q,m)$, $\daut$ deterministically reads the next input letter to be processed and stores it by moving to state~$(q,m,a)$. 
Now, this information, together with the top stack symbol~$X$ of the current configuration, suffices to determine the transition~$\lambda(m,a,X)$, which is then simulated by updating the state of $\aut$, the state of $\autm$, and the stack. 

Formally define the transitions and their colors as follows:
\begin{itemize}
    \item For every state~$(q, m) \in Q \times M$, every $X \in \Gamma_\bot$, and every $a \in \Sigma$, $\daut$ has the transition~$((q,m), X, a, (q,m,a), \gamma)$ of color~$\min \col(\Delta)$.
    \item For every state~$(q,m,a) \in Q \times M \times \Sigma$ and every $X \in \Gamma_\bot$ such that $\lambda(m,a,X)$ is an $\epsilon$-transition of the form~$\tau = (q, X, \epsilon, q', \gamma)$ for some $q'$ and some $\gamma$, $\daut$ has the transition~$((q,m,a), X, \epsilon, (q',\delta(m,\tau),a), \gamma)$ of color~$\col(\tau)$.
    \item For every state~$(q,m,a) \in Q \times M \times \Sigma$ and every $X \in \Gamma_\bot$ such that $\lambda(m,a,X)$ is an $a$-transition of the form~$\tau = (q, X, a, q', \gamma)$ for some $q'$ and some $\gamma$, $\daut$ has the transition~$((q,m,a), X, \epsilon, (q',\delta(m,\tau)), \gamma)$ of color~$\col(\tau)$.
\end{itemize}
By construction, there is exactly one enabled $a$-transition in every configuration with state in $Q \times M$ and there is at most one enabled $\epsilon$-transition in each configuration with state in $Q \times M \times \Sigma$. 
Hence, $\daut$ is indeed deterministic.
Finally, an induction shows that $L(\daut)$ is equal to $L(\aut)$ and therefore $L(\aut) \in$ \dcfl.
\end{proof}

Next, we consider resolvers implemented by pushdown automata with output.
Formally, a pushdown transducer (PDT) for an $\omega$-GFG-PDA~$\aut = (Q, \Sigma, \Gamma, q_\initmark, \Delta, \col)$ is a tuple~$\autt = (\mymark{Q}, \Delta, \mymark{\Gamma},\mymark{q_\initmark}, \mymark{\Delta}, \mymark{\lambda})$ such that $(\mymark{Q}, \Delta, \mymark{\Gamma}, \mymark{q_\initmark}, \mymark{\Delta})$ is a deterministic PDA (without coloring) processing transitions of $\aut$  and $\mymark{\lambda} \colon \mymark{Q} \times \Sigma \times \Gammabot \rightarrow \Delta $ is a (partial) output function.

A run of $\autt$ on an input~$v \in \Delta^*$ is a sequence~$c_0 \tau_0 c_1 \tau_1  \cdots \tau_{n}c_n $ such that $c_0 = (\mymark{q_\initmark}, \bot)$ and $c_i \trans{\tau_i} c_{i+1}$ for every $i < n$, $\ell(\tau_0 \cdots \tau_n) = v$, and $c_{n}$ has no enabled $\epsilon$-transitions.
Due to determinism and the requirement that the last configuration has no enabled $\epsilon$-transition, $\autt$ has at most one run on every input~$v$.

Let $\tau_0 \cdots \tau_n \in \Delta^*$ and $a \in \Sigma$. 
If the sequence~$\tau_0 \cdots \tau_n$ induces a run prefix of $\aut$, then let $X$ be the top stack symbol of its last configuration. 
In this case, we define $f_\autt(\tau_0 \cdots \tau_n, a) = \mymark{\lambda}(\mymark{q}, a, X)$, where $\mymark{q}$ is the state of the last configuration of the run of $\autt$ on $\tau_0 \cdots \tau_n$. 
If $\tau_0 \cdots \tau_n$ does not induce a run prefix of $\aut$, then we define $f_\autt(\tau_0 \cdots \tau_n, a)$ arbitrarily.
Again, we say that a resolver $\hstrat$ for $\aut$ is a pushdown resolver, if there is some PDT~$\autt$ for $\aut$ such that $\hstrat = f_\autt$. 

\begin{thm}
There is no pushdown resolver for the $\omega$-GFG-PDA~$\aut$ in Figure~\ref{fig:elaut}.
\end{thm}

\begin{proof}
We will proceed by contradiction. Assume $\hstrat$ is a pushdown resolver, implemented by some PDT~$\autt$, for the $\omega$-GFG-PDA~$\aut$ in Figure~\ref{fig:elaut} on Page~\pageref{fig:elaut}, which recognises $\seplang$.
Due to Remark~\ref{remark_runs}, every finite word has a run prefix processing it that is consistent with $\hstrat$.
Furthermore, since $\aut$ does not have any $\epsilon$-transitions this run is unique. 

Define the word $w$ by setting $w(0)$ to $\binom{-}{+}$ and $w(n)$ for $n>0$ to $\binom{-}{+}$ if the run induced by $r$ over $w(0)\cdots w(n-1)$ ends in state~$1$, and to $\binom{+}{-}$ otherwise. Observe that the run induced by $r$ over this word always has an empty stack and therefore the $r$-consistent run $\rho$ over $w$ is rejecting as it only uses transitions with color~$1$. Since $r$ is a resolver, it follows that  $w\notin \seplang$. Furthermore, since the stack remains empty throughout $\rho$, every position is a step in $\rho$, and the top stack symbol is always $\bot$ (see Page~\pageref{stepsdef} for the definition of steps).

Now, consider the run $\rho'$  of $\autt$ generating the run~$\rho$ processing $w$ in $\aut$. Since $\autt$ is a PDT, there are infinitely many steps along $\rho'$. Observe that if two positions $s$ and $s'$ are steps in $\rho'$ and agree on the top stack symbol in $\rho'$ and the state in both $\rho$ and $\rho'$, then the words $w(s)w(s+1)w(s+2)\cdots$ and $w(s')w(s'+1)w(s'+2)\cdots$ are identical: this is because $w(n+1)$ only depends on the transition chosen by $\autt$ after generating the run over $w(0)\cdots w(n)$, 
but $\autt$ is deterministic and at steps its behaviour does not depend on the stack configurations below the top stack symbols. Since both the stack alphabet of $\autt$ and state spaces of both $\aut$ and $\autt$ are finite, by the pigeonhole principle, such positions $s$ and $s'$ must exist. It follows that $w$ is ultimately periodic: $w=u v^\omega$ for some nonempty~$v$. 

We now argue that $w\in \seplang$. If $v$ has more occurrences of $\binom{+}{-}$ than $\binom{-}{+}$, then the first component has a safe suffix; similarly, if $v$ has more occurrences of $\binom{-}{+}$, then the second component has a safe suffix; finally, if there are as many of both letters, then both components have a safe suffix. In all cases, $w\in \seplang$, a contradiction.
We conclude that there is no pushdown resolver for $\aut$.
\end{proof}

This is perhaps surprising given~Theorem~\ref{thm:transducer} (in the appendix), which states that Player~$2$ has winning strategies implementable by pushdown transducers in Gale-Stewart games with \gfgcfl winning conditions. In particular, a consequence of Theorem~\ref{thm:transducer} is that a universal $\omega$-GFG-PDA~$\aut$ has a pushdown resolver. This can be seen by applying the construction presented in Section~\ref{sec_gfgaregfg} and in Appendix~\ref{section_appendixgames} to construct a winning strategy in the Gale-Stewart game with winning condition~$\set{\binom{w}{\#^\omega} \mid w\in L(\aut)}$, which is essentially a pushdown resolver for $\aut$.

%% file: content/acceptance.tex
In this section we study how $\omega$-GFG-PDA with different acceptance conditions compare in their expressivity. 
In particular, we establish that the classes of $\omega$-GFG-PDA with a fixed number of colors form an infinite hierarchy:
for each $n > 1$, there is a language recognized by an $\omega$-GFG-PDA with $n$ colors, but not recognized by an $\omega$-GFG-PDA with $n-1$ colors. 
This is also the case for $\omega$-DPDA~\cite{loding2015simplification} while $\omega$-PDA with Büchi acceptance (i.e., with colors $1$ and $2$) suffice to recognize all \cfl~\cite{DBLP:journals/jcss/CohenG77}.

\begin{lem}
For each $n > 1$, there is a language $L_n$ recognized by an $\omega$-GFG-PDA with $n$ colors but not recognized by any  $\omega$-GFG-PDA with $n-1$ colors.
\end{lem}

\begin{proof}
For each $n > 1$, let $D_n = \set{1,\ldots, n}$  and let $L_n \subseteq D_n^\omega$ be the language of $\omega$-words satisfying the parity condition, that is, in which the highest color that occurs infinitely often is even.
$L_n$ is recognized by a deterministic parity automaton with $n$ colors that upon reading a letter~$p$ visits a state of color~$p$. We can view this parity automaton as an $\omega$-GFG-PDA  (even an $\omega$-DPDA) that does not use its stack.

We now show that $L_n$ is not recognized by an $\omega$-GFG-PDA with only $n-1$ colors. Towards a contradiction, assume there is such an $\omega$-GFG-PDA~$\aut$ with resolver~$\hstrat$. The intuition of the argument is that we can build an infinite word $w$ of which the sequence of colors is (roughly) the sequence of colors on its $r$-consistent run, incremented by one, leading to a contradiction. Since we allow $\varepsilon$-transitions, we use the maximal color over a sequence of transitions to choose the next letter of $w$, rather than just the last transition. Without loss of generality, we can assume that the set of colors  of $\aut$ is either $\set{1,\ldots, n-1}$ or $\set{0,\ldots, n-2}$.
Consider the word $w(0)w(1)w(2)\cdots$ over $D_n$ built as follows: 
\begin{itemize}
    \item $w(0)=1$.
    \item For $i\geq 1$, first define $\rho_j$ for $j\geq i$ to be the shortest $\hstrat$-consistent run processing $w(0)\cdots w(j-1)$. Then, $w(i)=p+1$ where $p$ is the highest color in the suffix of $\rho_i$ starting after $\rho_{i-1}$ (or at the beginning of the run if $i=1$). In other words, $p$ is the maximal color in the run~$\rho_i$ on $w(0)\cdots w(i-1)$ after the transition processing $w(i-2)$. 
    
\end{itemize}

We now argue that the word $w$ defined as above induces an accepting $\hstrat$-consistent run if and only if it is not in $L_n$. To this end, define $p$ to be the highest color occurring infinitely often along the $\hstrat$-consistent run processing $w$, which exists due to Remark~\ref{remark_runs}. By construction, the highest color occurring infinitely often on $w$ is $p+1$. Therefore, the run is accepting if and only if $w$ is not in $L_n$, contradicting that $\aut$ recognizes $L_n$.
\end{proof}

The languages~$L_n$ used in the above proof are $\omega$-regular, and recognized by deterministic parity automata with $n$ colors, i.e., with the same number of colors as an $\omega$-GFG-PDA for $L_n$. 
However, by combining the languages~$L_n$ with $\seplang$, it is not hard to find a family of languages that requires both $n$ colors and an $\omega$-GFG-PDA.
Indeed, the language
\[\left\{\binom{w_1}{w_2} \,\middle|\, w_1\in \seplang \text{ and } w_2\in L_n\right\}\]
is recognized by an $\omega$-GFG-PDA  with $n$ colors, obtained by taking the product of an $\omega$-GFG-PDA for $\seplang$ with a deterministic parity automaton for $L_n$. 
However, it is not recognized by an $\omega$-DPDA, due to $\seplang$, nor by an $\omega$-GFG-PDA  with fewer than $n$ colors, due to $L_n$.

%% file: content/visibly.tex
In this section, we compare \gfgcfl to another important subclass of \cfl, the class of visibly pushdown languages~\cite{AlurM04}, for which solving games is decidable as well~\cite{DBLP:conf/fsttcs/LodingMS04}. 

Visibly pushdown automata are defined with respect to a partition~$\Sigmatilde = (\Sigmacall, \Sigmareturn, \Sigmaskip)$ of the input alphabet and have to satisfy the following conditions:
\begin{itemize}
    
    \item A letter~$a \in \Sigmacall$ is only processed by transitions of the form~$(q, X, a, q', XY)$ with $X\in \Gammabot$, i.e., some stack symbol~$Y$ is pushed onto the stack.
    
    \item A letter~$a \in \Sigmareturn$ is only processed by transitions of the form~$(q, X, a, q', \epsilon)$ with $X \neq \bot$ or $(q, \bot, a, q',\bot)$, i.e., the topmost stack symbol is removed, or if the stack is empty, it is left unchanged.
    
    \item A letter~$a \in \Sigmaskip$ is only processed by transitions of the form~$(q, X, a, q',X)$ with $X \in \Gammabot$, i.e., the stack is left unchanged.

    \item There are no $\epsilon$-transitions.
    
\end{itemize}
Intuitively, the stack height of the last configuration of a run processing some $v \in (\Sigmacall \cup \Sigmareturn \cup \Sigmaskip)^*$ only depends on $v$.

A language~$L\subseteq \Sigma^\omega$ is in \vpl if there is a partition~$\Sigmatilde$ of $\Sigma$ such that there is a nondeterministic~$\omega$-visibly pushdown automaton~$\aut$ recognizing $L$ with respect to $\Sigmatilde$. 

\begin{thm}
\label{thm_vpl}
\gfgcfl and \vpl are incomparable with respect to inclusion.
\end{thm}

\begin{proof}
The language~$\seplang \in$ \gfgcfl used in Section~\ref{sec_gfg} to separate \gfgcfl and \dcfl is not in \vpl, as \vpl $\subseteq$ \cfl is closed under complementation~\cite{AlurM04} while the complement of $\seplang$ is not in \cfl (Theorem~\ref{thm_closure}).
Thus, \gfgcfl is not included in \vpl.

For the other non-inclusion, let $\Sigma = \set{\plus, \minus }$ and define the value\footnote{Alur and Madhusudan used the term~\textit{stack height} instead of \textit{value}, but this is misleading here, since our automata are not necessarily visibly, i.e., the stack height of a run prefix on $v$ might differ from $\val(v)$.}~$\val(v) \in \nats$ of a finite word~$v \in \Sigma^*$ inductively as $\val(\epsilon) = 0$ as well as $\val(v\plus) = \val(v) +1$ and $\val(v\minus) = \max\set{0, \val(v)-1}$. This corresponds to the height of a stack after the sequence of push and pop operations described by $v$ when $\plus$ is interpreted as a push and $\minus$ as a pop. We consider the language of such sequences that visit some value infinitely often.
We show that
\begin{align*}
\repbddlang =  \{& w \in \Sigma^\omega \mid  \text{there is an } s\in\nats \text{ such that }\\
& \val(w(0) \cdots w(n)) = s \text{ for infinitely many $n$}\}
,    
\end{align*}
which is in \vpl~\cite{AlurM04}, is not in \gfgcfl. 

To this end, fix an $\omega$-GFG-PDA~$\aut$ with resolver~$\hstrat \colon \Delta^* \times\Sigma \rightarrow \Delta$. We will show that it does not recognize $\repbddlang$.
We assume without loss of generality that all $\epsilon$-transitions have color~$0$ while all $\Sigma$-transitions have a non\-zero color.
This can be achieved by adding a component to $\aut$'s states that accumulates the maximal color seen along a sequence of $\epsilon$-transitions until a $\Sigma$-transition is used.
As a consequence, a run of $\aut$ on some infinite input satisfies the acceptance condition if and only if the sequence of $\Sigma$-transitions satisfies the acceptance condition, i.e., the colors of $\epsilon$-transitions are irrelevant and will be ignored in the following.

Given a run prefix~$\rho \in \Delta^*$ and a letter~$a \in \Sigma$, let $\ext(\rho, a)$ be the unique extension of $\rho$ induced by $\hstrat$ when processing~$a$. 
Formally, we define
\[
\ext(\rho, a) = \begin{cases}
\rho\cdot\hstrat(\rho, a) & \text{if }\ell(\hstrat(\rho, a)) = a,\\
\ext(\rho \cdot \hstrat(\rho, a),a)&\text{if }\ell(\hstrat(\rho, a)) = \epsilon.
\end{cases}
\]
Due to Remark~\ref{remark_runs}, if $\rho$ is a run prefix that is consistent with $\hstrat$ then $\ext(\rho, a)$ is a finite extension of $\rho$.

We now build an $\omega$-word~$w$, letter by letter, based on how $\hstrat$ resolves nondeterminism on the prefix built so far. The intuition is that whenever the run built by $\hstrat$ sees an even color after a prefix~$v$, the value of prefixes extending $v$ remains above $\val(v)$ until a larger odd color is seen, and then returns to $\val(v)$ (unless a higher even color is seen in the meantime). Roughly, after an even color occurs in the run, the value of the word increases until a higher odd color occurs. Thus, if the maximal color that occurs infinitely often is even, i.e., the run is accepting, then the word will not be in the language.
On the other hand, when an odd color $p$ occurs, the value of the run decreases to just above the value of the word at the last even color higher than $p$. Hence if the maximal color that occurs infinitely often is odd, i.e., if the run is rejecting, then the value reaches the same level infinitely often, and the word is in the language.  The result is that either $w$ is in $\repbddlang$ but rejected by $\aut$, or $w$ is not in $\repbddlang$ but accepted by $\aut$.

Formally, we inductively define an infinite sequence of sequences~$\rho_n \in \Delta^*$, all ending in a $\Sigma$-transition.
To start, we define $\rho_0 = \ext(\ext(\epsilon, \plus),\plus)$.
To define $\rho_{n+1}$ we have to consider several cases.

First, assume $\rho_n$ ends with a $\plus$-transition.
Let $p$ be the last and $p'$ be the second-to-last nonzero color appearing in $\rho_n$ (this is well-defined as $\rho_0$ contains two colors). 
If $p$ is odd and $p \ge p'$, then we define
$\rho_{n+1} = \ext(\rho_n, \minus)$ (Case~1),
otherwise, $\rho_{n+1} = \ext(\rho_n, \plus)$ (Case~2).

Now, assume $\rho_n$ ends with a $\minus$-transition.
Let $\suffix{\rho_n}$ be the suffix of $\rho_n$ starting with the last $\plus$-transition (this is well-defined, as $\rho_0$ contains $\plus$-transitions).
Furthermore, let $\prefix{\rho_n}$ be the prefix of $\rho_n$ ending with the last transition having an even color that is at least as large as the maximal color labeling a transition in $\suffix{\rho_n}$.
See Figure~\ref{fig:caseminus} for an illustration: $\prefix{\rho_n}$ is the prefix ending in the last transition that has an even color~$p$ that is at least as large as the colors~$p_0,\ldots, p_j$, the colors occurring in $\suffix{\rho_n}$.
If there is no transition with such a color, then define $\prefix{\rho_n} = \epsilon$. 
Note that $\suffix{\rho_n}$ and $\prefix{\rho_n}$ might overlap and that $\prefix{\rho_n} = \rho_n$ is possible if the last transition of $\rho_n$ has an even color that is maximal among those in $\suffix{\rho_n}$.
If 
\[\val(\ell(\rho_n)) > \val(\ell(\prefix{\rho_n})) + 1,\] then we define
$\rho_{n+1} = \ext(\rho_n, \minus)$ (Case~3),
otherwise, $\rho_{n+1} = \ext(\rho_n, \plus)$ (Case~4).
Note that the even-numbered cases extend by a $\plus$, the odd-numbered cases by a $\minus$.
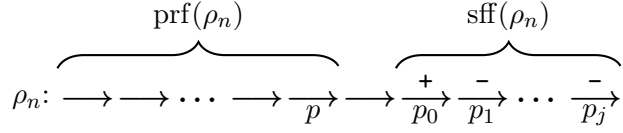
\begin{figure} 
\centering
    \begin{tikzpicture}[thick]
    
\node at (0,0) {$\rho_n$:};

\foreach \x in {1,2,4,5,...,8,10}{
\node at (.75*\x,0) {{\boldmath$\longrightarrow$}};
}
\node at (.75*3,0) {{\boldmath$\cdots$}};
\node at (.75*9,0) {{\boldmath$\cdots$}};

\foreach \x in {8,10}{
\node at ({.75*\x-.00},.24) {{\boldmath$\minus$}};
}
\node at ({.75*7-.00},.24) {{\boldmath$\plus$}};

\node at ({.75*5-.00},-.2) {{$p$}};

\foreach \x in {0,1} {
\node at ({.75*\x+7*.75},-.2) {{$p_{\x}$}};
}
\node at ({.75*3+7*.75},-.2) {{$p_{j}$}};

\draw [decorate,decoration={brace,amplitude=10pt},xshift=10pt,yshift=12pt]
(.75*6,0) -- (.75*10,0) node [black,midway,yshift=0.7cm] 
{$\suffix{\rho_n}$};

\draw [decorate,decoration={brace,amplitude=10pt},xshift=10pt,yshift=12pt]
(.75*0,0) -- (.75*5,0) node [black,midway,yshift=0.7cm] 
{$\prefix{\rho_n}$};

    \end{tikzpicture}  
    \caption{Illustration of the definition of $\rho_{n+1}$ in the case where $\rho_n$ ends with a $\minus$-transition. Each arrow depicts a transition of $\rho_n$, its color is depicted below the arrow, the letter it processes above (some transitions in $\suffix{\rho_n}$ may be $\epsilon$-transitions, but only the first one is a $\plus$-transition).}
    \label{fig:caseminus}
\end{figure}

Now, let $\rho \in \Delta^\omega$ be the unique infinite sequence such that each $\rho_n$ is a prefix of $\rho$, and let $w = \ell(\rho)$, which is an infinite word. 
Then, $\rho$ induces a run of $\aut$ on $w$ that is consistent with $\hstrat$. 

First, assume $w$ is of the form~$v\plus^\omega$. 
Then, from some point onwards, we only use Case~2 to extend $\rho_n$ to $\rho_{n+1}$, i.e., if the last color of $\rho_n$ is odd, then the second-to-last color is strictly larger. 
This implies that the maximal color occurring infinitely often in $\rho$ is even. 
Thus, the run $\rho$ is accepting although $w$ is not in $\repbddlang$.

Now, assume $w$ is of the form~$v\minus^\omega$, i.e., 
$\val(\ell(\rho_n)) = 0$
for almost all $n$.
Then, from some point onwards, we only use Case~3 to extend $\rho_n$ to $\rho_{n+1}$, i.e., we have
\[\val(\ell(\rho_n)) > \val(\ell(\prefix{\rho_n})) + 1\]
for almost all $n$.
Combining both equations yields a contradiction, as $\val(\ell(\prefix{\rho_n})) + 1$ is positive.
Thus, $w$ cannot have the form~$v\minus^\omega$.

As a last case, assume $w$ contains infinitely many $\plus$ and infinitely many $\minus$. 
First, we study the case where the maximal color occurring infinitely often in $\rho$, call it $p$, is odd.
Let $n_0$ be such that the suffix of $\rho$ obtained from removing $\rho_{n_0}$ only contains colors that occur infinitely often and $\rho_{n_0}$ contains at least one $p$ that is not followed by a larger even color. 
Furthermore, let $\rho^*$ be the longest prefix of $\rho$ ending in a transition with an even color that is larger than $p$ (which by construction is a prefix of $\rho_{n_0}$).
We define $\rho^* = \epsilon$ if there is no such color. 
Let $b = \val(\rho^*)$.
We claim that there are infinitely many $n$ such that $\val(\ell(\rho_n)) \le b +1$.
Then, $\rho$ is rejecting while $w$ is in $\repbddlang$ due to the pigeonhole principle.

To this end, let $n' > n_0$ such that $\rho_{n'}$ ends with a transition~$\tau$ of color~$p$.
As there are infinitely many such $n'$, it suffices to show that there is a $n \ge n'$ such that $\val(\ell(\rho_n)) \le b +1$.
First, assume $\ell(\tau) = \plus$. 
By construction, the last nonzero color occurring before the final $p$ is not larger than $p$.
Hence, we have $\rho_{n'+1} = \ext(\rho_{n'}, \minus)$ due to Case~1.
Now, we either already have $\val(\rho_{n'+1}) \le b +1 $ or we apply Case~3 repeatedly until we have produced some $\rho_{n}$ with $\val(\rho_{n}) = b +1$.
The reason why Case~3 is always applicable is that the suffix~$\suffix{\rho_{n''}}$ for $n'+1 \le n'' \le n$ always contains the last transition of $\rho_{n'}$, with color~$p$.
This in turn implies that the prefix~$\prefix{\rho_{n''}}$ is equal to $\rho^*$.

The case for $\ell(\tau) = \minus$ is similar:
either we already have $\val(\rho_{n'}) \le b  +1$ or we apply Case~3 repeatedly until we have produced some $\rho_{n}$ with $\val(\rho_{n}) = b  +1$.

Finally, we consider the case where~$p$, the maximal color occurring infinitely often in $\rho$, is even.
We show that there are infinitely many $n$ such that $\val(\ell(\rho_{n'} )) > \val(\ell(\rho_n))$ for every $n' > n$. 
This implies that for every $s$ there are only finitely many $n$ such that $\val(\ell(\rho_{n})) = s$. 
As the prefixes of $w$ are of the form~$\ell(\rho_{n})$ for $n \in \nats$, we obtain that $w$ is not in $\repbddlang$, even though the run of $\aut$ on $w$ induced by $\rho$ is accepting.

Let $n_0$ be such that the suffix of $\rho$ obtained from removing $\rho_{n_0}$ only contains colors that occur infinitely often and $\rho_{n_0}$ contains a suffix~$\rho_s$ starting with a transition of color~$p$ such that $\rho_s$ does not contain a larger color than $p$, but does contain a $\plus$-transition.
By definition, the resulting suffix of $\rho$ has infinitely many transitions of color~$p$, all of which mark the end of some~$\rho_n$ with $n \ge n_0$.
We show that each such $n$ has the desired property.

By the choice of $n_0$, the last transition of $\rho_{n+1}$ is a $\plus$-tran\-si\-tion, no matter whether the last transition of $\rho_n$ is a $\plus$-transition or a $\minus$-transition.
If it is a $\plus$-transition, then $\rho_{n+1}$ is obtained by applying Case~2 to $\rho_n$, as the color of the last transition of $\rho_n$ is even.
On the other hand, if the last transition of $\rho_n$ is a $\minus$-transition, then $\rho_{n+1}$ is obtained by applying Case~4 to $\rho_n$, as the prefix~$\prefix{\rho_n}$ is equal to $\rho_n$ by the fact that $\rho_{n_0}$ contains an occurrence of a $p$ that is not followed by a larger color, but by a $\plus$-transition. 

Furthermore, $\rho_{n+2}$ is obtained by applying Case~2 to $\rho_{n+1}$, as the color of the last transition of $\rho_{n+1}$ might be odd, but then it is strictly smaller than $p$, the second-to-last color in $\rho_{n+1}$. Therefore, the last transition of $\rho_{n+1}$ is also a $\plus$-transition.

Now, assume towards a contradiction that there is some $n' > n$ such that $\val(\ell(\rho_{n'})) = \val(\ell(\rho_n))$. 
Pick $n'$ minimal with this property.  Then, $n'>n+2$ since the last transitions of both $\rho_{n+1}$ and $\rho_{n+2}$ are $\plus$-transitions. Therefore, due to minimality, the last transition of $\rho_{n'}$ and the last transition of $\rho_{n'-1}$ are both $\minus$-transitions. 
Hence, $\rho_{n'} = \ext(\rho_{n'-1},\minus)$ due to Case~3.

Now, consider the prefix~$\prefix{\rho_{n'-1}}$ in the application of Case~3 to $\rho_{n'-1}$. 
It is either equal to $\rho_{n}$ (as its color is even and at least as large as all colors that may be appear in the suffix~$\suffix{\rho_{n'-1}}$ of $\rho_{n'-1}$), or it is some extension of $\rho_{n}$.
In both cases, we have $\val(\ell(\rho_n)) \le \val(\ell(\prefix{\rho_{n'-1}}))$ (in the latter due to the minimality of $n'$).
Hence,
\begin{align*}
\val(\ell(\rho_n)) \le{}&{} \val(\ell(\prefix{\rho_{n'}-1}))< \val(\ell(\rho_{n'-1})) - 1 \\ ={}&{} \val(\ell(\rho_{n'})) = \val(\ell(\rho_{n})),
\end{align*}
which yields the desired contradiction. 
Here, the strict inequality follows from the definition of Case~3.

To conclude, in either of the cases we considered,  the run~$\rho$ is accepting, but $w \notin \repbddlang$, or the run $\rho$ induced by the resolver~$\hstrat$ is rejecting, but $w \in \repbddlang$. 
Hence, either $\aut$ does not recognize $\repbddlang$ or $\hstrat $ is not a resolver for $\aut$, i.e., $\repbddlang$ is not in \gfgcfl.
\end{proof}

%% file: content/conclusion.tex
We have introduced good-for-games $\omega$-pushdown automata and proved that they recognize a novel class of $\omega$-context\-free languages for which solving games (and synthesizing winning strategies) is decidable.
Furthermore, we have studied (the mostly nonexistent) closure properties of the new class, proven that it is incomparable to $\omega$-visibly pushdown languages, and that deciding good-for-gameness is undecidable for $\omega$-pushdown automata and $\omega$-contextfree languages.
Finally, we proved the parity index hierarchy to be strict for $\omega$-GFG-PDA and studied resource-bounded resolvers: finite-state resolvers only exist for \dcfl and even pushdown resolvers are not sufficient for $\omega$-GFG-PDA.

We hope this paper is the catalyst for an in-depth study of good-for-games automata in settings where nondeterministic automata are more expressive than deterministic ones.
But even for the setting of $\omega$-contextfree languages considered here, we open many interesting directions for further research.
Let us conclude by listing a few:
\begin{itemize}
    \item Is the universality problem for $\omega$-GFG-PDA $\exptime$-complete?  Note that this problem is a promise problem with an undecidable promise, e.g., we only consider $\omega$-GFG-PDA as input.
    
    \item Equivalence is decidable for $\omega$-DPDA with weak acceptance conditions~\cite{equivalenceWeak} (that is, each strongly connected component is either rejecting or accepting), but it is undecidable for $\omega$-DPDA with Büchi or co-Büchi acceptance conditions~\cite{equivalenceBuchi}. 
    Whether weak acceptance conditions make the problem decidable for $\omega$-GFG-PDA is, to the best of our knowledge, open.
    
    \item Can \gfgcfl be characterized by some extension of Monadic Second-order Logic? Such characterizations haven been exhibited for contextfree languages of finite~\cite{Lautemann} and infinite words~\cite{Droste} as well as for the class of visibly pushdown languages~\cite{AlurM04}. However, there is, to the best of our knowledge, no characterization of the deterministic contextfree languages, neither for finite nor infinite words.
    
    \item We have shown that there are $\omega$-GFG-PDA that do not have pushdown resolvers, but left open whether for every $\omega$-GFG-CFL there is an $\omega$-GFG-PDA recognizing that language with a pushdown resolver. In fact it is also open whether every $\omega$-GFG-PDA or every \gfgcfl has a computable resolver.
    
    \item While here we focus on infinite words, good-for-games pushdown automata turn out to be more powerful than their deterministic counterparts, already over finite words. In~\cite{gfgfin}, we study their expressivity and succinctness over finite words, but some gaps remain open. 

\item Another interesting direction, proposed by one of the reviewers, is to further restrict the model of $\omega$-GFG-PDA, by considering weaker classes of nondeterministic $\omega$-PDA, e.g., unambiguous ones or one-counter automata. For finite words, unambiguity and good-for-gameness are incomparable~\cite{gfgfin}.
\end{itemize}

\subsubsection*{Acknowledgements} We thank Sven Schewe and Patrick Totzke for fruitful discussions leading to the results presented in Section~\ref{section_decisionproblems}.
{\lettrine[image=true, lines=2, findent=1ex, nindent=0ex, loversize=.12]{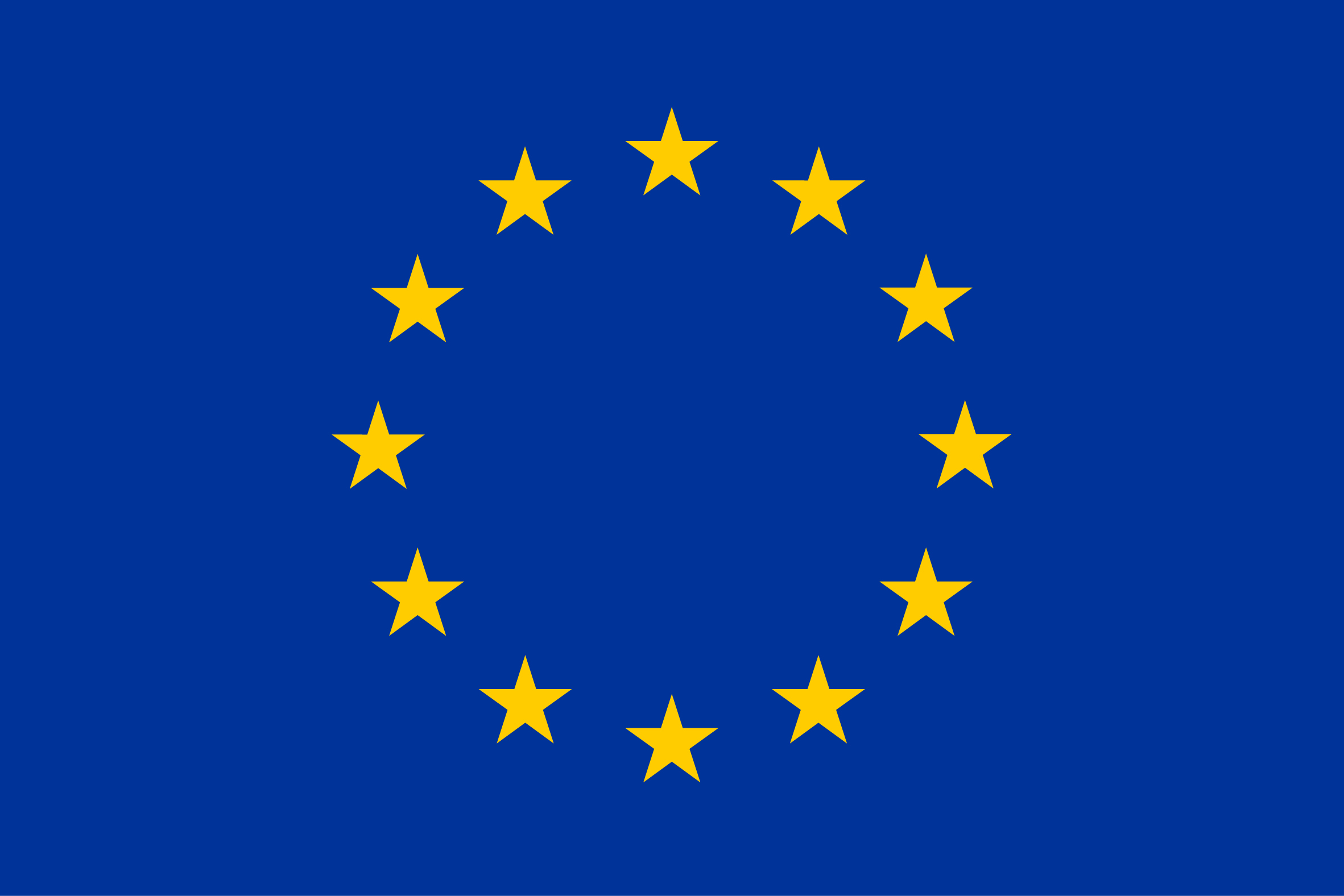} This work is part of a project that receives funding from the European Union’s Horizon 2020 research and innovation programme under the Marie Skłodowska-Curie grant
agreement No~892704.\par}

%% file: content/appendix.tex
Let $h\colon \set{0,1,\#}^*\rightarrow \set{0,1}^* $ be the homomorphism induced by $h(0) = 0$, $h(1)=1$, and $h(\#)=\epsilon$.
Define 
\[P = \set{ v\#^\omega \mid h(v)  = x \rev{x} \text{ for some } x \in \set{0,1}^*},\]
where $\rev{x}$ denotes the reversal of $x$.
It is straightforward to construct an $\omega$-PDA recognizing $P$, thereby showing $P \in$ \cfl. 
We show $P \notin $ \gfgcfl. 

Towards a contradiction, let us assume that $P$ is recognized by an $\omega$-GFG-PDA~$\aut = (Q, \Sigma, \Gamma, q_\initmark, \Delta, \col)$ with resolver~$\hstrat\colon \Delta^* \times \Sigma \rightarrow \Delta$.
Note that $\hstrat$ induces a unique run $\rho_w$ for every $w \in P$. 
Furthermore, if two words~$w,w' \in P$ share a prefix, then their runs~$\rho_w = c _0 \tau_0 c_1 \tau_1 c_2 \tau_2 \cdots$ and $\rho_{w'} = c_0' \tau_0' c_1'
\tau_1' c_2' \tau_2' \cdots$ also share a prefix.
More formally, if 
\[w(0) \cdots w(n) = w'(0) \cdots w'(n)\] 
for some $n \ge 0$, then we have $\tau_0 \cdots \tau_k = \tau_0' \cdots \tau_k'$ and thus $c_0 \cdots c_{k+1} = c_0' \cdots c_{k+1}'$ for every $ k$ such that $\ell(\tau_0 \cdots \tau_k) = w(0) \cdots w(n-1)$.
Here, the $-1$ stems from the fact that $\hstrat$ resolves nondeterminism based on the next input letter to be processed.

Let $\rho = c_0 \tau_0 c_1 \cdots c_k t_k c_{k+1} $ be an alternating sequence of configurations and transitions,  and let  $a \in \set{0,1,\#}$. 
We say that $(\rho, a)$ is $\hstrat$-consistent, if 
\[
\tau_n = \hstrat( \tau_0 \cdots \tau_{n-1}, v( \size{ \ell(\tau_0 \cdots \tau_{n-1}) } ) )
\]
and $c_n$ is the configuration reached by $\aut$ after the transitions~$\tau_0 \cdots \tau_{n-1}$, where $v =\ell(\rho) a$.
Now, an $\hstrat$-consistent pair~$(\rho, a)$ has Property~M if the following holds for every $w \in \set{0,1,\#}^*\#^\omega$:
Let $ c_0' \tau_0' c_1' \tau_1' c_2' \tau_2' \cdots$ be the unique run of $\aut$ on $\ell(\rho)aw$ induced by $\hstrat$.
Then, we require $\sh(c_{k'+1}) \ge \sh(c_{k+1})$ for all $k' > k$, i.e., that $k+1$ is a step (see Page~\pageref{stepsdef}). 
Note that we have $ \tau_0' \cdots \tau_k' = \tau_0 \cdots \tau_k$ and $c_0 \cdots c_{k+1} = c_0' \cdots c_{k+1}'$ by definition.

Lifting Property~M to words, we say that $v(0) \cdots v(n) \in \set{0,1,\#}^+$ has property~M if there is a run prefix~$\rho$ with $\ell(\rho) = v(0) \cdots v(n-1)$ such that $(\rho, v(n))$ is $\hstrat$-consistent and has Property~M.

We show that for every $v \in \set{0,1,\#}^+$ there is a $v' \in \set{0,1,\#}^*$ such that $vv'$ has Property~M.
To this end, assume $v(0) \cdots v(n)$ does not have Property~M.
Let $\rho$ be a run prefix of $\aut$ on $v(0) \cdots v(n-1)$ so that $(\rho, v(n))$ is $\hstrat$-consistent. 
As $v$ is a prefix of some word in $P$, such a $\rho$ exists. 
By definition, $(\rho, a)$ does not have property~M, i.e., there is a run induced by $\hstrat$ processing a prolongation in $v \set{0,1,\#}^*\#^\omega$ that reaches a stack height strictly smaller than the stack height of the last configuration of $\rho$. 
As stack heights are bounded from below, there is a minimal stack height that is assumed by such runs. 
Let $vv'$ with $v' \in \set{0,1,\#}^*$ be a word processed by such a run prefix~$\rho'$ ending with a minimal stack height along all runs considered. 
Then, $vv'0$ has Property~M, as $(\rho', 0)$ has property~M, as after $\rho'$ no strictly smaller stack height is reached by runs induced by $\hstrat$.

Now, fix $n = 3\size{Q}(\size{\Gamma +1}) + 1$ and define $v_j = 01^j0$ for every $j$ in the range~$1 \le j \le n$.
As shown above, for every such $j$, there is a $v_j'$ such that $v_jv_j'$ has Property~M.

Now, each $v_jv_j' \rev{(v_jv_j')}\#^\omega$ is in $L$, i.e., the unique run
\[c_0^j \tau_0^j c_1^j \tau_1^j c_2^j \tau_2^j \cdots\]
of $\aut$ on $v_j v_j' \rev{(v_jv_j')}\#^\omega$ induced by $\hstrat$ is accepting.
As each $v_jv_j'$ has Property~M, there are prefixes~$c_0^j \tau_0^j \cdots c_{k_j}^j \tau_{k_j}^j c_{k_j+1}^j$ of $c_0^j \tau_0^j c_1^j \tau_1^j c_2^j \tau_2^j \cdots$ processing $v_jv_j'$ without its last letter~$a_j$ such that $(c_0^j \tau_0^j \cdots c_{k_j}^j \tau_{k_j}^j c_{k_j+1}^j, a_j)$ has Property~M.

Now, there are $j_0 \neq j_1$ such that the configurations $c_{k_{j_0}+1}^{j_0}$ and $c_{k_{j_1}+1}^{j_1}$ coincide on their state from $Q$ and their top stack symbol from $\Gamma$ and such that $a_{j_0} = a_{j_1}$. 

Consider the sequence
\[\rho^* = \tau_0^{j_0} \cdots \tau_{k_{j_0}}^{j_0} \tau_{k_{j_1}+1}^{j_1} \tau_{k_{j_1}+2}^{j_1} \tau_{k_{j_1}+3}^{j_1} \cdots .
\] 
We claim that $\rho^*$ induces an accepting run of $\aut$ on $\overline{w} =  v_{j_0}v_{j_0}' \rev{(v_{j_1}v_{j_1}')}\#^\omega  $. 
We have
\[
v_{j_0}v_{j_0}' \rev{(v_{j_1}v_{j_1}')} =  01^{j_0}0 v_{j_0}' \rev{(v_{j_1})} 01^{j_1}0  ,
\]
which is not of the form $v\rev{v}$ after removing $\#$'s. 
Hence, this step completes the proof, as $\overline{w} \notin P$ is accepted by $\aut$, yielding the desired contradiction. 

The sequence $\rho^*$ satisfies the acceptance condition, as it shares a suffix  with the sequence of transitions of the accepting run~$c_0^{j_1} \tau_0^{j_1} c_1^{j_1} \tau_1^{j_1} c_2^{j_1}\tau_2^{j_1} \cdots$.
Furthermore, we have 
\[
\ell(\rho^*) = \ell(\tau_0^{j_0} \cdots \tau_{k_{j_0}}^{j_0} ) \ell( \tau_{k_{j_1}+1}^{j_1} \tau_{k_{j_1}+2}^{j_1} \tau_{k_{j_1}+3}^{j_1} \cdots)
\]
where $\ell(\tau_0^{j_0} \cdots \tau_{k_{j_0}}^{j_0})$ is equal to $v_{j_0} v_{j_0}'$ without its last letter~$a_{j_0}$, and where \[\ell( \tau_{k_{j_1}+1}^{j_1} \tau_{k_{j_1}+2}^{j_1} \tau_{k_{j_1}+3}^{j_1} \cdots) = a_{j_1} \rev{(v_{j_1}v_{j_1}')}\#^\omega .\] 
The concatenation of these two words is indeed $\overline{w}$, as we have $ a_{j_0} = a_{j_1} $ by construction. 

Finally, as $(c_0^{j_1} \tau_0^{j_1} \cdots   c_{k_{j_1}}^{j_1}\tau_{k_{j_1}}^{j_1} c_{k_{j_1}+1}^{j_1}, a_{j_1}) $ has Property~M, the run $c_0^{j_1}\tau_0^{j_1}c_1^{j_1} \tau_1^{j_1} c_2^{j_1} \tau_2^{j_1} \cdots$ does after position~$k_{j_1}$ not depend on the complete stack content at that position, but only on the state and the top stack symbol of $c_{k_{j_1}+1}$. 
Hence, appending the run suffix resulting from applying the transitions
\[
\tau_{k_{j_1}+1}^{j_1} \tau_{k_{j_1}+2}^{j_1} \tau_{k_{j_1}+3}^{j_1} \cdots
\]
after
$(c_0^{j_0} \tau_0^{j_0} \cdots   c_{k_{j_0}}^{j_0}\tau_{k_{j_0}}^{j_0} c_{k_{j_0}+1}^{j_0}, a_{j_0}) $  (recall that $c_{k_{j_0}+1}^{j_0}$ and $c_{k_{j_1}+1}^{j_1}$  have the same state and top stack symbol) yields indeed a run of $\aut$ (but not necessarily induced by $\hstrat$).

%% file: content/gamesappendix.tex
In Section~\ref{sec_gfgaregfg}, we showed that the winner of a Gale-Stewart game with an $\omega$-GFG-CFL winning condition can be determined in exponential time. 
Here, we consider the synthesis problem: If Player~$2$ wins a Gale-Stewart game with an $\omega$-GFG-CFL winning condition, compute a (finitely represented) winning strategy for her.
We show that this problem can be solved in exponential time as well, as for Gale-Stewart games with $\omega$-DCFL winning conditions.

Assume we are given an $\omega$-GFG-PDA~$\aut$ such that Player~$2$ wins $\gsgame(L(\aut))$. 
In the proof of Theorem~\ref{thm_gfggfg}, we have constructed a polynomially-sized $\omega$-DPDA~$\autud$ such that Player~$2$ wins $\gsgame(L(\aut))$ if and only if she wins $\gsgame(L(\autud))$.
We show how a finitely represented winning strategy for Player~$2$ in $\gsgame(L(\autud))$, which exists if she wins $\gsgame(L(\autud))$, can be turned into the desired strategy for $\gsgame(L(\aut))$.

In our context, we need a slight adaption of the pushdown transducer model we introduced in Section~\ref{section_resolver} to implement resolvers.
Here, a pushdown transducer (PDT) is a tuple~$\autt = (\mymark{Q}, \Sigma_I, \mymark{\Gamma},\mymark{q_\initmark}, \mymark{\Delta}, \Sigma_O, \mymark{\lambda})$ such that $(\mymark{Q}, \Sigma_I, \mymark{\Gamma}, \mymark{q_\initmark}, \mymark{\Delta})$ is a deterministic PDA (without coloring) over $\Sigma_I$, $\Sigma_O$ is an output alphabet, and $\mymark{\lambda} \colon \mymark{Q} \rightarrow \Sigma_O $ is a (partial) output function.
Note that we only change the signature of the output function. 
Hence, runs of $\autt$ are defined as before.
In particular, we still require them to be maximal, i.e., they end in configurations without enabled $\epsilon$-transitions.
Thus, we say that $\autt$ computes the (partial) function~$\sigma_\autt \colon \Sigma_I^* \rightarrow \Sigma_O$ given by $\sigma_\autt(w) = \mymark{\lambda}(q)$, where $q$ is the state of the configuration the run of $\autt$ on $w$ ends in.

In a Gale-Stewart game~$\gsgame(L)$ with $L \subseteq (\SigmaI \times \SigmaO)^\omega$, a strategy for Player~$2$ is a function from $\SigmaI^+$ to $\SigmaO$. 
    Such strategies can be implemented by a pushdown transducer with input alphabet~$\Sigma_I = \SigmaI$ and output alphabet~$\Sigma_O = \SigmaO$.

\begin{propC}[\cite{DBLP:journals/corr/abs-1006-1415,DBLP:journals/iandc/Walukiewicz01}]
\label{prop:transducer}
Let $\autud$ be an $\omega$-DPDA.
If Player~$2$ wins $\gsgame(L(\autud))$, then she has a winning strategy that is implemented by a PDT. Furthermore, such a PDT can be computed in exponential time in $\size{\autud}$.
\end{propC}

%
%

Now, we are ready to prove Theorem~\ref{thm:transducer}.

\begin{proof}
Given an $\omega$-GFG-PDA~$\aut$, let the $\omega$-DPDA~$\autud$ be defined as in the proof of Theorem~\ref{thm_gfggfg}. 
Then, the game~$\gsgame(L(\autud))$ is won by Player~$2$ and she has a winning strategy for $G(\autud)$ that is implemented by a PDT, which can be computed in exponential time (in the size of $\aut$, as $\autud$ is only polynomially larger than $\aut$).
Now, we show that the transformation $\sigma \mapsto \sigma_{-d}$ presented in the proof of Theorem~\ref{thm_gfggfg} is effective, if $\sigma$ is implemented by a PDT. That is, we turn a PDT~$\autt$ implementing $\sigma$ into a PDT~$\autt_{-d}$ implementing~$\sigma_{-d}$.

\newcommand{\iphase}{\iota}
\newcommand{\wphase}{\omega}

So, consider a PDT~$\autt = (\mymark{Q},(\SigmaI)_d,  \mymark{\Gamma},\mymark{q_\initmark}, \mymark{\Delta}, (\SigmaO)_d, \mymark{\lambda})$ implementing a winning strategy~$\sigma$ for Player~$2$ in $\gsgame(\autud)$
where we write $(\SigmaI)_d$ for $\SigmaI$ and write $(\SigmaO)_d$ for $\SigmaO \cup \Delta$ to distinguish the alphabets of $\aut$ and $\autud$.
The mode of a configuration~$(q, \gamma X)$ of $\autt$ is the pair~$(q, X)$ (note that $X$ might be $\bot$). 
Whether a transition is enabled in a configuration only depends on its mode.
Thus, we are justified to say that a mode enables a transition or not.
Now, a mode is a reading mode if it enables some non-$\epsilon$ transition.
As $\autt$ is deterministic, no reading mode has an enabled $\epsilon$-transition.
Furthermore, due to the maximality requirement on runs, every run ends in a reading mode. 
Finally, our construction relies on the following simple fact: If we remove all but one transition enabled by a reading mode, and turn the remaining transition into an $\epsilon$-transition, then the resulting PDT is still deterministic.

First, we transform $\autt$ into a PDT computing the same strategy as $\autt$ while keeping track of its own mode using its states.
This is useful as it allows us to base the output function, which formally only depends on a state of a configuration, on the current mode.
Thus, the new PDT has to keep track of the topmost stack symbol, which is straightforward during transitions pushing a symbol on the stack and during transitions replacing the topmost stack symbol, but requires an additional transition to read the newly revealed stack symbol after the topmost stack symbol is popped off the stack.

Formally, define $\autt' = (\mymark{Q} \times \mymark{\Gamma}_\bot \cup Q, (\SigmaI)_d, \mymark{\Gamma},(\mymark{q_\initmark}, \bot), \mymark{\Delta}',  (\SigmaO)_d, \mymark{\lambda}')$ with
\begin{align*}
\mymark{\Delta}' = 
& \set{  ((q,X), X, a, (q',X''), X'X'') \mid  (q, X, a, q', X'X'') \in \mymark{\Delta}  } \cup\\
& \set{ ((q,X), X, a, (q', X'),X') \mid  (q, X, a, q', X') \in \mymark{\Delta} }\cup\\
& \set{ ((q,X), X, a, q', \epsilon) \mid  (q, X, a, q', \epsilon) \in \mymark{\Delta} } \cup\\
& \set{ (q, X, \epsilon, (q,X), X) \mid q\in \mymark{Q}, X \in \mymark{\Gamma}_\bot  },
\end{align*}
which is still deterministic.
Further, $\autt'$ has a run on some input~$v$ if and only if $\autt$ has a run on $v$.
Finally, let the run of $\autt$ on $v$ end in some configuration~$(q, \gamma X)$.
Then, the run of $\autt'$ on $v$ ends in the configuration~$((q,X), \gamma X)$.
Hence, both PDT's indeed compute the same function when defining $\mymark{\lambda}'(q, X) = \mymark{\lambda}(q)$.

Now, we define $\autt_{-d} = (Q_{-d},\SigmaI, \mymark{\Gamma},(\mymark{q_\initmark}, \bot,\iphase), \mymark{\Delta}_{-d},  \SigmaO, \mymark{\lambda}_{-d})$ with state set
\[Q_{-d} = (\mymark{Q} \times \mymark{\Gamma}_\bot \cup Q) \times (\set{\iphase, \wphase} \cup \SigmaO).\]
Next, we define the transition relation~$\mymark{\Delta}_{-d}$, whose definition refers to the fixed dummy letter~$\dummy \in \SigmaI$ used in the definition of $\sigma_{-d}$ (see Page~\ref{page:dummyletter}).
Intuitively, $\autt_{-d}$ works in three distinct phases as follows:
\begin{itemize}
    \item The initialization phase spans the initial sequence of $\epsilon$-transitions executed by $\autt'$ and is left for the waiting phase as soon as the first letter from $\SigmaI$ is processed. Formally,
       \begin{itemize}
           \item $((q,\iphase), X, \epsilon, (q',\iphase), \gamma) \in \mymark{\Delta}_{-d}$ for every $(q, X, \epsilon, q', \gamma) \in \mymark{\Delta}'$, and
           \item $((q,\iphase), X, a, (q',\wphase), \gamma) \in \mymark{\Delta}_{-d}$ for every $(q, X, a, q', \gamma) \in \mymark{\Delta}'$ with $a \in \SigmaI$.
       \end{itemize}
    \item In the waiting phase, $\epsilon$-transitions are simulated until none are enabled any more. At that point, $\aut'$ yields an output and processes some input from $\SigmaI$. In $\autt_{-d}$, the output letter is stored and the unique transition labeled by $\dummy$ is turned into an $\epsilon$-transition, which starts the delay phase. As we remove all other $a$-transitions with $a\neq \dummy$, this preserves determinism.
    Formally, 
        \begin{itemize}
           \item $((q,\wphase), X, \epsilon, (q',\wphase), \gamma) \in \mymark{\Delta}_{-d}$ for every $(q, X, \epsilon, q', \gamma) \in \mymark{\Delta}'$ and
           \item $((q,\wphase), X, \epsilon, (q',\mymark{\lambda}'(q)), \gamma) \in \mymark{\Delta}_{-d}$ for every $(q, X, \dummy, q', \gamma) \in \mymark{\Delta}'$.
       \end{itemize}
    \item During the delay phase, $\epsilon$- and $\dummy$-transitions of $\autt'$ are simulated (deterministically, and without processing an input letter) until a reading mode is reached at which $\autt'$ outputs a non-$\epsilon$-transition, signifying the end of a block. At that point, the output letter stored during the delay phase is finally output and the run returns to the waiting phase. As we remove all other transitions when turning $\dummy$-transitions into $\epsilon$-transitions, we again preserve determinism. Formally,
        \begin{itemize}
           \item $((q,a_2), X, \epsilon, (q',a_2), \gamma) \in \mymark{\Delta}_{-d}$ for every  $a_2 \in \SigmaO$ and every $(q, X, \epsilon, q', \gamma) \in \mymark{\Delta}'$,
           \item $((q,a_2), X, \epsilon, (q',a_2), \gamma) \in \mymark{\Delta}_{-d}$ for every  $a_2 \in \SigmaO$ and every $(q, X, \dummy, q', \gamma) \in \mymark{\Delta}'$ such that $q$ is a reading mode of $\autt$ and $\mymark{\lambda}'(q)$ is an $\epsilon$-transition of $\aut$, and 
           \item $((q,a_2), X, a_1, (q',\wphase), \gamma) \in \mymark{\Delta}_{-d}$ for every  $a_2 \in \SigmaO$ and every $(q, X, a_1, q', \gamma) \in \mymark{\Delta}'$ such that $q$ is a reading mode of $\autt$, but $\mymark{\lambda}'(q)$ is a non-$\epsilon$-transition of $\aut$.
       \end{itemize}
\end{itemize}
Finally, we define $\mymark{\lambda}_{-d}(q,a_2) = a_2$ for every state~$q$ of $\autt'$ that is a reading mode of $\autt$, but $\mymark{\lambda}'(q)$ is a non-$\epsilon$-transition of $\aut$.
Thus, we indeed output the stored letter at the end of the delay phase as described above. 

Recall that in the proof of Theorem~\ref{thm_gfggfg}, we have turned the winning strategy~$\sigma$ for Player~$2$ in $\gsgame(L(\autud))$ into a winning strategy $\sigma_{-d}$ for her in $\gsgame(L(\aut))$. 
To this end, we first defined a function mapping $v \in \SigmaI^+$ to some $v_d \in (\SigmaI')^+$ as follows: 
$a_d =a$ for $a \in \SigmaI$ and $(va)_d = v_d\dummy^n a$ for some $n$ that is independent of $a$.  
With these definitions, we defined $\sigma_{-d}(v) = \sigma(v_d)$.
In the remainder of the proof, we show that $\autt_{-d}$ indeed computes $\sigma_{-d}$. 

So, consider some $v \in \SigmaI^+$ and the unique $n$ such that $(va)_d = v_d \dummy^n a$ for all $a \in \SigmaI$. 
Then, we have $\sigma(v_d) = \mymark{\lambda}'(q)$, where $q$ is the state of the last configuration of the run~$\rho$ of $\autt'$ on $v_d$.
Furthermore, consider the run~$\rho'$ of $\autt'$ on $v_d\dummy^n$, which is a proper extension of $\rho$ ending in some state~$q'$.
By definition of $n$, $q'$ is a reading mode of $\autt$ and $\mymark{\lambda}'(q')$ is a non-$\epsilon$-transition of $\aut$. 

The run~$\rho'$ can be turned into the unique run of $\autt_{-d}$ on $v$ by switching between the phases appropriately:
\begin{itemize}
    \item The initialization phase ends with the transition processing the first letter of $v_d$, which starts a waiting phase.
    \item Each waiting phase lasts until a state is reached at which a letter in $\SigmaO$ is output by $\autt'$. This letter is stored and a delay phase is started.
    \item Each delay phase lasts until a state is reached at which a non-$\epsilon$-transition is output by $\autt'$. 
\end{itemize}
In particular, after the last letter of $v$ is processed, a waiting phase starts, which is ended by switching from state~$(q,\wphase)$ to a state storing~$\mymark{\lambda}'(q)$. The last delay phase stores this letter until the end of the run, at which point it yields the output~$\mymark{\lambda}'(q)$ of $\autt_{-d}$ on $v$. Thus, $\autt_{-d}$ indeed implements~$\sigma_{-d}$.
\end{proof}

Finally, let us consider the situation of Player~$1$ in games with \gfgcfl winning conditions, where we define strategies and strategies implemented by PDT's for Player~$1$ as expected. 
As \gfgcfl is not closed under complementation, we cannot dualize a game with a \gfgcfl winning condition by swapping the roles of the players and complementing the winning condition to obtain a winning strategy for Player~$1$.
Instead, the situation is asymmetric.

\begin{thm}
There is a Gale-Stewart game with an \gfgcfl winning condition that is won by Player~$1$, but not with a winning strategy implementable by a PDT.
\end{thm}

\begin{proof}
In Section~\ref{sec_gfg}, we have introduced the language~$\seplang$ over $I \times I$ for $I = \set{\plus, \minus,\zero}$
 and the words $x_1 = \binom{\plus}{\zero} \binom{\plus}{\minus}$ and $x_2 = \binom{\zero}{\plus} \binom{\minus}{\plus}$.
Now, consider the game~$\gsgame(L)$ with
\[
L = \left\{ \binom{w}{\#^\omega} \,\middle|\, w \in \seplang \text{ or } w \notin \set{x_1, x_2}^\omega \right\},
\]
which is in \gfgcfl due to Theorem~\ref{theorem_closurereg}.

Recall that the $\omega$-word
\[w_{\overline{ss}} =
x_1 \, 
(x_2)^3 \,
(x_1)^7 \,
(x_2)^{15} \,
(x_1)^{31} \,
(x_2)^{63} \,
\cdots 
\]
defined on Page~\pageref{bla} is not in $\seplang$.
Thus, Player~$1$ wins $\gsgame(L)$ by producing the word~$w_{\overline{ss}}$.

Now, consider a PDT~$\autt$ implementing a strategy~$\sigma$ for Player~$1$ in $\gsgame(L)$.
Let $w\binom{w}{\#^\omega}$ be an outcome that is consistent with~$\sigma$.
The run~$\rho$ of $\autt$ generating this outcome has infinitely many steps.
Now, pick two such steps, such that the configurations at these steps coincide on their state, top stack symbol, and such that at least three outputs are generated between the steps.
Then, the run obtained by repeating the sequence of transitions taken between the steps generates an outcome consistent with $\sigma$.
Now, this outcome has a positive energy level in one component, as every infix of length at least~$3$ of a word built by concatenating copies of the $x_i$ has a strictly positive energy level in one component.
Hence, $\sigma$ is not a winning strategy for Player~$1$ in $\gsgame(L)$.

Hence, no PDT implements a winning strategy for Player~$1$ in $\gsgame(L)$.
\end{proof}

Note that Player~$1$ \emph{does} have a winning strategy implementable by a PDT in the game~$\gsgame(L(\autud))$ with \dcfl winning condition obtained from $\gsgame(L)$ (see the proof of Theorem~\ref{thm_gfggfg}).
In the game~$\gsgame(L(\autud))$, Player~$2$ has to construct a run of an $\omega$-GFG-PDA recognizing $L$, which introduces enough information for Player~$1$ to win with a strategy implementable by a PDT.
However, in $\gsgame(L)$, he does not have access to that additional information, and therefore no winning strategy implementable by a PDT.